\documentclass[a4paper,11pt,twoside]{article}
\usepackage[a4paper,left=2.2cm,right=2.2cm,top=2.3cm,bottom=2.6cm]{geometry}
\usepackage[english]{babel}
\usepackage{setspace}
\usepackage{mathrsfs,hyperref, algorithm, algorithmic}
\usepackage{harvard}
\usepackage[]{amsmath}
\usepackage{amsthm}
\usepackage{fix-cm}
\usepackage[]{amssymb}
\usepackage[]{latexsym}
\usepackage[latin1]{inputenc}
\usepackage[right]{eurosym}
\usepackage[T1]{fontenc}
\usepackage[]{graphicx}
\usepackage[]{epsfig}
\usepackage{fancyhdr}
\usepackage{bbm}
\usepackage{pstricks}
\usepackage{multirow}
\usepackage{rotating}
\usepackage{array}
\usepackage{slashbox}
\makeatletter

\renewcommand{\p@enumi}{theenumi-}
\renewcommand{\@fnsymbol}[1]{\@alph{#1}}

\newcommand{\bbr}{\mathbb{R}}

\newcommand{\bbn}{\mathbb{N}}

\newcommand{\fn}[1]{\footnote{\protect\doublespacing #1}}
\newcommand{\ci}{\citeasnoun}

\newcommand{\pcal}{\mathcal{P}}

\newcommand{\ga}{\alpha}
\newcommand{\gb}{\beta}

\newcounter{modcount}
\newcommand{\modulo}[2]{%
\setcounter{modcount}{#1}\relax
\ifnum\value{modcount}<#2\relax
\else\relax
\addtocounter{modcount}{-#2}\relax
\modulo{\value{modcount}}{#2}\relax
\fi}

\newcommand{\tablepictures}[4][c]{\begin{tabular}[#1]{@{}c@{}}#2\vspace{0.5cm}\\(\alph{#4}) #3\end{tabular}}
\newcounter{gridsearch}
\newcommand{\tabpic}[2]{
    \stepcounter{gridsearch}
    \modulo{\thegridsearch}{2}
    \ifnum\value{modcount}=0
        \tablepictures[t]{#1}{#2}{gridsearch}\\[2.0cm]
    \else
        \tablepictures[t]{#1}{#2}{gridsearch}&~&
    \fi
}

\makeatother
\newtheorem{lemma}{Lemma}[section]

\newtheorem{theorem}[lemma]{Theorem}

\newtheorem{example1}[lemma]{Example}
\newtheorem{ex1}[lemma]{Example}
\newtheorem{rem1}[lemma]{Remark}
\newtheorem{assumption}[lemma]{Assumption}
\newtheorem{alg1}[lemma]{Algorithm}
\newtheorem{me1}[lemma]{Mechanism}

\newenvironment{rem}{\begin{rem1}\rm}{\end{rem1}}

\newenvironment{example}{\begin{example1}\rm}{\end{example1}}

\usepackage{color}

\newcommand{\T}{\mathsf{T}}

\DeclareMathOperator*{\argmax}{arg\,max}

\numberwithin{figure}{section}
\numberwithin{table}{section}

\begin{document}

\title{The Effects of Leverage Requirements and Fire Sales on Financial Contagion via Asset Liquidation Strategies in Financial Networks}
\author{ Zachary Feinstein\fn{Zachary Feinstein, ESE, Washington University, St. Louis, MO 63130, {\tt zfeinstein@ese.wustl.edu}.}\\[0.7ex] \textit{Washington University in St. Louis} \and Fatena El-Masri\fn{Fatena El-Masri, Federal Deposit Insurance Corporation, Arlington, VA 22203, {\tt fatenaelmasri@gmail.com}}\\[0.7ex] \textit{Federal Deposit Insurance Corporation}}
\date{\today\fn{Opinions expressed in this paper are those of the authors and not necessarily those of the FDIC. The authors are grateful to the editors and referees for their thoughtful comments and encouragements that led to this greatly improved paper.}}
\maketitle

\begin{abstract}
This paper provides a framework for modeling the financial system with multiple illiquid assets when liquidation of illiquid assets is caused by failure to meet a leverage requirement.  This extends the network model of \ci*{CFS05} which incorporates a single asset with fire sales and capital adequacy ratio.  This also extends the network model of \ci*{feinstein2015illiquid} which incorporates multiple illiquid assets with fire sales and no leverage ratios.  We prove existence of equilibrium clearing payments and liquidation prices for a known liquidation strategy when leverage requirements are required.  We also prove sufficient conditions for the existence of an equilibrium liquidation strategy with corresponding clearing payments and liquidation prices.  Finally we calibrate network models to asset and liability data for 50 banks in the United States from 2007-2014 in order to draw conclusions on systemic risk as a function of leverage requirements.
\end{abstract}\vspace{0.2cm}

\section{Introduction}

Leverage or capital adequacy requirements are regulatory rules to constrain the risk of financial institutions.  The more levered a firm, the greater the impact of an adverse shock to its balance sheet.  However, financial institutions do not exist in isolation; the actions of one firm can impact the health of other banks.  This effect is known as \emph{financial contagion}.  There are many avenues for contagion to spread including local interactions (e.g., contractual obligations) and global interactions (e.g., price impact caused by deleveraging and the mark-to-market accounting rules).  The risk to the financial system, as opposed to the risk to a specific bank, is known as \emph{systemic risk}.  In this paper, we will consider an extension of the financial contagion model of \ci*{EN01} to include multiple illiquid assets with fire sales triggered so as to have leverage ratios below the mandated maximum. This model is similar to that presented in \ci*{feinstein2015illiquid}, though with the additional allowance for fire sales caused by leverage rather than a capital shortfall.  As noted in \ci*{CFS05}, regulation that is intelligent for individual firms may actually exacerbate a systemic crisis via mark-to-market valuation from forced liquidations.  This model is further motivated by \ci*{GY14} insofar as financial contagion is weak under the model of \ci*{EN01} without asset liquidations.  Thus accurate modeling of the entire financial system is of paramount importance.

\ci*{EN01} presents a model of interbank liabilities and studies how defaults spread through the financial system due to unpaid financial obligations.  In that paper existence and uniqueness of the clearing payments are proved and algorithms for the computation of the payment vector resultant from the propagation of defaults.  
This financial contagion model has been generalized to include, e.g., bankruptcy costs (cf.\ \ci*{E07}, \ci*{RV13}, \ci*{EGJ14}, \ci*{GY14}, and \ci*{AW_15}), cross-holdings (cf.\ \ci*{E07}, \ci*{EGJ14}, and \ci*{AW_15}), fire sales for a single (representative) illiquid asset (cf.\ \ci*{CFS05}, \ci*{NYYA07}, \ci*{GK10}, \ci*{AFM13}, \ci*{CLY14}, \ci*{AW_15}, and \ci*{AFM16}), and fire sales for multiple illiquid assets (cf.\ \ci*{feinstein2015illiquid}).  \ci*{CSMF14} and \ci*{GLT15} consider network models of financial institutions in which connections are defined through common asset holdings explicitly rather than financial obligations.  Outside of a network model, \ci*{CW13} and \ci*{CW14} develop a framework for modeling correlations of asset prices during a fire sale.  Additionally, there have been empirical studies of systemic risk models in, e.g., \ci*{ELS06}, \ci*{U11}, \ci*{CMS10}, and \ci*{GY14}.

As an extension of \ci*{feinstein2015illiquid}, we study the contagion model with multiple illiquid assets where a fire sale is triggered if the leverage of a firm goes above some required level.  This is in comparison to the model presented in \ci*{CFS05} with a single illiquid asset that is liquidated if a capital adequacy requirement is violated.  We provide results on existence of clearing payments and asset prices under
\begin{enumerate}
\item known liquidation strategies and
\item an equilibrium, i.e., valuation maximizing, liquidation strategy for each financial institution.
\end{enumerate}
We additionally study how our leverage model behaves with an empirical study of financial institutions in the United States.  These case studies allow us to consider counterfactual scenarios and consider the outcome if leverage requirements were set to change instantaneously.  These numerical studies demonstrate that contagion is primarily driven by asset liquidation, as studied in \ci*{GY14}.

The organization of this paper is as follows: Section~\ref{Sec:setting} provides the financial setting and general insights that we will work within for the remainder of the paper.  Section~\ref{Sec:clearing} provides a discussion of the network clearing mechanism, and resultant interbank payments and asset prices, when the liquidation strategies are dependent only on the state of model and not the actions of other firms.  Section~\ref{Sec:games} provides a discussion of the network clearing mechanism under a game theoretic framework in which each firm is a (mark-to-market) value maximizer.  Throughout we emphasize conditions for the existence of an equilibrium solution to these clearing problems.  Section~\ref{Sec:cases} considers numerical simulations to demonstrate some insights that can be gained by our leveraged network model.  We focus on network models which we calibrate to banking data in the United States from 2007-2014.

\section{Setting} \label{Sec:setting}

Consider the financial system presented in \ci*{feinstein2015illiquid}.  That is, there are $n$ financial institutions (e.g., banks, hedge funds, or pension plans) with bilateral obligations -- commonly this is depicted as a directed network of nominal liabilities.  Further, there is a financial market with a num\'{e}raire -- a liquid (cash) asset -- and $m$ illiquid assets whose prices fluctuate depending on the numbers of units being bought or sold by the financial institutions.  For the purposes of this paper, all liabilities are to be paid in the cash asset.  We denote by $p \in \bbr^n_+$ the realized clearing payments of the banks that are attained based on the bilateral obligations.  We denote by $q \in \bbr^m_+$ the prices of the illiquid assets that are attained based on the liquidations that the financial firms are forced to enact.
Throughout this paper we will use the notation 
\begin{align*}
x \wedge y &= (\min(x_1,y_1),\min(x_2,y_2),...,\min(x_d,y_d))^\T,\\
x \vee y &= (\max(x_1,y_1),\max(x_2,y_2),...,\max(x_d,y_d))^\T
\end{align*} 
where $x,y \in \bbr^d$ for some $d \in \bbn$.

As described in \ci*{EN01}, bilateral obligations mean that any financial institution $i \in \{1,2,...,n\}$ may be a creditor to other firms, likewise institution $i$ may be an obligor to other firms as well.  Though central counterparty clearing houses (CCPs) are not within the scope of this paper, it is possible to extend the network in such a way as described in \ci*{AFM13}.  Let $\bar p_{ij} \geq 0$ be the nominal liability of firm $i$ to firm $j$.  We will impose the restriction that no firm has an obligation to itself, that is, $\bar p_{ii} = 0$ for every $i$.
The \emph{total liabilities} of firm $i$ are given by
\[\bar p_i := \sum_{j = 1}^n \bar p_{ij}.\]
To ease notation, we will define $\bar p \in \bbr^n_+$ to be the vector of total obligations for each firm in the financial system.
The fractional amount of total liabilities that firm $i$ owes to firm $j$, called the \emph{relative liabilities}, is given by
\[a_{ij} = \begin{cases}\frac{\bar{p}_{ij}}{\bar{p}_i} & \text{if } \bar{p}_i > 0 \\ \frac{1}{n} & \text{else}\end{cases}.\]
By construction the relative liabilities satisfy $\sum_{j = 1}^n a_{ij} = 1$ for any $i$.  The relative liabilities matrix is defined as $A = (a_{ij})_{i,j = 1,2,...,n}$. Note that $a_{ij}$ can be chosen arbitrarily when $\bar p_i = 0$; the choice $a_{ij} = \frac{1}{n}$ in such a case is to guarantee the summation is equal to $1$.  If sufficient marked-to-market wealth is not available then a financial institution will default on their obligations.  We accept the assumption that there is no priority of payments, as presented in \ci*{EN01}. That is, payments will be made based on the relative liabilities matrix $A$.

Along with obligations, each firm $i \in \{1,2,...,n\}$ has an initial endowment of cash $x_i \geq 0$ and some number $s_i \in \bbr^m_+$ of each illiquid asset.  The illiquid assets are measured in physical units, i.e., firm $i$ holds $s_{ik} \geq 0$ units of asset $k \in \{1,2,...,m\}$ rather than denoting by the valuation.  In this way, we can describe the vector of liquid endowments by $x \in \bbr^n_+$ and the matrix of illiquid endowments by $S = (s_{ik}) \in \bbr^{n \times m}_+$.

As mentioned above, due to forced liquidation of the illiquid assets the price of these assets may be shocked downward.  In particular this downward shock is dependent on the clearing mechanism and not external to the financial system.  We model this price by a vector valued \emph{inverse demand function} $F: \bbr^m_+ \to [0,\bar q] \subseteq \bbr^m_+$ for some maximum prices $\bar q \in \bbr^m_{++}$.  The inverse demand function $F$ maps the number of units of each illiquid asset being sold into the corresponding price per unit that would be quoted in the financial market.  Of particular interest is the component-wise inverse demand function, i.e.,
\[F(\gamma) := \left(\hat F_1(\gamma_1),\hat F_2(\gamma_2),...,\hat F_m(\gamma_m)\right)^\T \quad (\forall \gamma \in \bbr^m_+)\]
for scalar valued inverse demand functions $\hat F_k: \bbr_+ \to [0,\bar q_k]$ for every $k = 1,2,...,m$.  Mathematically, however, we will consider the more general setting in which the liquidation of one asset may potentially influence the prices of the other assets as well.  This general setting permits us to consider the case with correlated prices during a crisis as studied in \ci*{CW13} and \ci*{CW14}.  For the remainder of this paper we assume the following conditions on the inverse demand function (as was done in \ci*{feinstein2015illiquid} as well).
\begin{assumption}\label{Ass:idf}
The inverse demand function $F: \bbr^m_+ \to [0,\bar q]$ is continuous and nonincreasing.
\end{assumption}

We are now able to construct the general rules that the clearing payments and liquidations must follow.  This setting is comparable to those in \ci*{CFS05} and \ci*{feinstein2015illiquid}.  For the valuation of the illiquid assets, we will follow mark-to-market accounting.  This means the value of firm $i$'s endowments, under price vector $q \in \bbr^m_+$ is given by
\[x_i + q^\T s_i = x_i + \sum_{k = 1}^m s_{ik}q_k.\]
Beyond the endowments, each firm $i$ receives payments from firm $j$.  As described above, in case of default by firm $j$, the payments are proportional to the size of the obligation, i.e., $p_{ji} = a_{ji} p_j$ if firm $j$ pays $p_j \geq 0$ into the system as a whole.  Therefore, after taking the obligations and liabilities into account, the wealth of firm $i$ is given by
\[x_i + \sum_{k = 1}^m s_{ik}q_k + \sum_{j = 1}^n a_{ji} p_j - p_i.\]
As in \ci*{EN01}, we assume limited liability for all obligations.  That is, no firm will go into debt to pay its obligations.  Mathematically this implies that the wealth of any firm must be at least equal to $0$.  Further, we assume that each firm must pay off all its debt before it can accrue positive wealth.  Thus, under pricing vector $q \in \bbr^m_+$, we immediately can conclude that firm $i$ will pay out
\[p_i = \bar{p}_i \wedge \left(x_i + \sum_{k = 1}^m s_{ik} q_k + \sum_{j = 1}^n a_{ji} p_j\right).\]
That is, firm $i$ will pay the minimum of its total obligations ($\bar p_i$) and its marked-to-market assets ($x_i + \sum_{k = 1}^m s_{ik} q_k + \sum_{j = 1}^n a_{ji} p_j$).

At any point the leverage ratio for firm $i$ is given by
\[\lambda_i = \frac{\bar p_i - \left(t_i + \sum_{k = 1}^m \gamma_{ik} q_k\right)}{x_i + \sum_{k = 1}^m s_{ik} q_k + \sum_{j = 1}^n a_{ji} p_j - p_i}\]
where $t_i$ is the amount of the liquid asset ``sold'' and $\gamma_{ik}$ is the amount of illiquid asset $k$ sold by firm $i$.  The leverage requirement is the constraint that $\lambda_i \leq \lambda_i^{\max} \in \bbr_+$.
The numerator of the leverage ratio is given by the total debt (after liquidations).  The denominator is the equity of firm $i$ after realized clearing occurs.
We will impose two constraints on the liquidations.  First, we assume that the market does not allow for short-selling, i.e., $t_i \leq x_i + \sum_{j = 1}^n a_{ji} p_j$ and $\gamma_{ik} \leq s_{ik}$ for every firm $i$ and asset $k$.  Second, we assume that no firms within the financial network will purchase assets during a fire sale, i.e., $t_i \geq 0$ and $\gamma_{ik} \geq 0$.
Under the weak assumption that liquid capital will be used to pay off liabilities before illiquid assets are sold we can conclude
\begin{align*}
t_i &= \bar p_i \wedge \left(x_i + \sum_{j = 1}^n a_{ji} p_j\right)\\
q^\T \gamma_i &\geq \left(\bar p_i + \lambda^{\max}_i \left[p_i - \left(x_i + \sum_{k = 1}^m s_{ik} q_k + \sum_{j = 1}^n a_{ji} p_j\right)\right] - t_i\right)^+
\end{align*}

\begin{rem}
If $\lambda^{\max}_i = 0$ for all firms $i$, the model presented herein reduces to the setting of \ci*{feinstein2015illiquid}.
\end{rem}

As in \ci*{feinstein2015illiquid}, no further conclusions can be drawn without a discussion of how the firms will choose to liquidate their holdings when forced via the mechanism above.  We will consider two settings for this: when a known closed-form liquidation strategy is chosen (Section \ref{Sec:clearing}) and when a game theoretic wealth optimizing equilibrium strategy is chosen (Section \ref{Sec:games}).

\section{Clearing mechanism under known liquidation strategy} \label{Sec:clearing}

Consider the setting presented above in Section \ref{Sec:setting}.  In this section, we study the existence of the clearing payments and clearing asset prices when the liquidation strategy for each firm is known and \emph{not} dependent on the liquidation strategy chosen by the other firms.  As in \ci*{feinstein2015illiquid}, we define the \emph{liquidation function} $\gamma_{ik}: [0,\bar p] \times [0,\bar q] \to \bbr_+$ to be the number of units of asset $k \in \{1,2,...,m\}$ that firm $i \in \{1,2,...,n\}$ wishes to liquidate.  For notational simplicity, we will denote the vector of liquidations for firm $i$ by
\[\gamma_i(p,q) = \left(\gamma_{i1}(p,q),\gamma_{i2}(p,q),...,\gamma_{im}(p,q)\right)^\T \in \bbr^m_+\]
for clearing payments $p \in [0,\bar p]$ and asset prices $q \in [0,\bar q]$.  Additionally, we will denote $\gamma(p,q) \in \bbr^{n \times m}_+$ to be the matrix of asset liquidations for the financial system under clearing payments $p \in [0,\bar p]$ and asset prices $q \in [0,\bar q]$.

Recall that we assume there exists a no short-selling constraint for the financial market.  Therefore firm $i$ can only sell $s_{ik} \wedge \gamma_{ik}(p,q)$ units of asset $k$ given the clearing vector $p$ and price vector $q$.  Given the desired liquidations $s_i \wedge \gamma_i(p,q)$ of every firm $i$, the prices in the market may be updated due to the price impact, i.e., $\hat q = F\left(\sum_{i = 1}^n \left[s_i \wedge \gamma_i(p,q)\right]\right)$ is the updated market prices.

To describe the clearing mechanism with a given liquidation function $\gamma$, we use the previously given equity and pricing formulations.  The \emph{clearing mechanism} $\phi: [0,\bar p] \times [0,\bar q] \to [0,\bar p] \times [0,\bar q]$ gives the updated payment and pricing vectors.  By the logic discussed above, we define the clearing mechanism pointwise for any $(p,q) \in [0,\bar p] \times [0,\bar q]$ by
\begin{equation}\label{Eq:clearing}
\phi(p,q) := \left(\begin{array}{c}\bar{p} \wedge \left(x + S q + A^\T p\right) \\ F\left(\sum_{i = 1}^n \left[s_i \wedge \gamma_i(p,q)\right]\right)\end{array}\right).
\end{equation}
The fixed points of the clearing mechanism, i.e.,
\[(p^*,q^*) = \phi(p^*,q^*),\]
are the \emph{realized payment} or \emph{clearing vector} $p^* \in [0,\bar p]$ and implied \emph{clearing price vector} $q^* \in [0,\bar q]$ of the illiquid assets.

Because of the inverse demand function is nonincreasing, to maximize valuation, firms would want to sell the fewest assets necessary to satisfy the leverage requirements.  This is encoded in the following \emph{minimal liquidation} condition on the liquidation function $\gamma$.
\begin{assumption}\label{Ass:min-liquidation}
The liquidation function $\gamma: [0,\bar p] \times [0,\bar q] \to \bbr^{n \times m}_+$ satisfies the minimal liquidation condition: 
\begin{align}
\label{Eq:min-liquidation} q^\T\left[s_i \wedge \gamma_i(p,q)\right] &=\\
\nonumber (q^\T s_i) &\wedge \left(\left[\bar p_i - \left(x_i + \sum_{j = 1}^n a_{ji} p_j \right)\right]^+ - \lambda^{\max}_i \left[\left(x_i + \sum_{k = 1}^m s_{ik} q_k + \sum_{j = 1}^n a_{ji} p_j\right) - \bar p_i\right]^+\right)^+
\end{align}
for all firms $i \in \{1,2,...,n\}$.
\end{assumption}
As described in \ci*{feinstein2015illiquid}, Assumption~\ref{Ass:min-liquidation} states that the number of units liquidated by each firm (under clearing vector $p$ and prices $q$) either is enough to fulfill the leverage requirement or all assets are liquidated.  Additionally, there are no extraneous liquidations, i.e., no firm sells more units than is necessary to satisfy the leverage requirement.

\begin{theorem}\label{Thm:clearing-exist}
Consider a financial system $(A,\bar{p})$ with liquid endowments $x \in \bbr^n_+$ and illiquid endowments $S \in \bbr^{n \times m}_+$.  Consider liquidation function $\gamma$ and inverse demand function $F$ satisfying Assumption~\ref{Ass:idf}.
\begin{enumerate}
\item If the liquidation function $\gamma$ is continuous, then there exists a clearing payment and pricing vector $(p^*,q^*)$.
\item If the liquidation function $\gamma$ is nonincreasing, then there exists a greatest and least clearing payment vector and vector of prices, $(p^+,q^+) \geq (p^-,q^-)$.
\end{enumerate}
\end{theorem}
\begin{proof}
This follows identically to Theorem~3.6 of \ci*{feinstein2015illiquid}.
\end{proof}

We will now give a few examples of liquidation functions.

\begin{example}\label{Ex:1stock-gamma}
In a market with a single (representative) illiquid asset, i.e., $m = 1$, the minimum liquidation constraint \eqref{Eq:min-liquidation} implies 
\[\gamma_i(p,q) = \frac{1}{q}\left(\left[\bar p_i - \left(x_i + \sum_{j = 1}^n a_{ji} p_j\right)\right]^+ - \lambda^{\max}_i \left[\left(x_i + q s_i + \sum_{j = 1}^n a_{ji} p_j\right) - \bar p_i\right]^+\right)^+.\]
This is akin to the single asset model presented in, e.g., \ci*{CFS05}, though is a distinct result as the leverage requirements given herein are different from the capital adequacy requirements in \ci*{CFS05}.
\end{example}

\begin{example}\label{Ex:proportional-gamma}
In a market with multiple illiquid assets, firms may choose to sell their assets in proportion to their holdings, i.e., for some agent $i = 1,2,...,n$ and any asset $k = 1,2,...,m$
\[\begin{split}\gamma_{ik}(p,q) &=\\
&\frac{s_{ik}}{\sum_{l = 1}^m s_{il}q_l}\left(\left[\bar p_i - \left(x_i + \sum_{j = 1}^n a_{ji} p_j\right)\right]^+ - \lambda^{\max}_i \left[\left(x_i + \sum_{k = 1}^m s_{ik} q_k + \sum_{j = 1}^n a_{ji} p_j\right) - \bar p_i\right]^+\right)^+.
\end{split}\]
This is the leveraged extension of Example~3.3 in \ci*{feinstein2015illiquid}.  The specified number of units to be sold is equal to the fraction (leverage shortfall divided by the marked-to-market asset valuation) of the total holdings $s_{ik}$.
\end{example}

\begin{example}\label{Ex:price-gamma-best}
To decrease the total number of assets to be liquidated, firms may choose to sell the assets with greatest value first.  Notationally, let $[k] \in \{1,2,...,m\}$ indicate the index of the $k^\text{th}$ highest price, i.e., $q_{[1]} \geq q_{[2]} \geq ... \geq q_{[m]}$.  The number of units of asset $[k]$ liquidated by firm $i$ is given by 
\[\begin{split}\gamma_{i[k]}&(p,q) =\\
&\frac{1}{q_{[k]}}\left(\left[\bar p_i - \left(x_i + \sum_{j = 1}^n a_{ji} p_j\right)\right]^+ - \lambda^{\max}_i \left[\left(x_i + \sum_{l = 1}^m s_{il} q_l + \sum_{j = 1}^n a_{ji} p_j\right) - \bar p_i\right]^+ - \sum_{l = 1}^{k-1} s_{i[l]} q_{[l]}\right)^+.
\end{split}\]
Notice that, in the case that a firm is able to cover the leveraged shortfall through the priciest $k-1$ assets alone, then no units of assets $l = k,k+1,...,m$ would be sold.
However, we may not be able to guarantee the existence of a solution under this liquidation strategy.  Instead we propose an strategy that approximates the behavior described above.  For any $\epsilon > 0$ define a function $h_\epsilon: \bbr_+ \to [0,1]$ continuous and strictly decreasing.  Additionally define $h_\epsilon$ so that
\[\lim_{\epsilon \to 0} h_\epsilon(z) = \begin{cases} 1 & \text{if } z = 0\\ 0 &\text{else}\end{cases}.\]  
As described we can define the liquidation strategy via the following system of equations.
\begin{align}
\label{Eq:price-gamma-best}
\gamma_{i[k]}^\epsilon(p,q) &= \sum_{\hat k = 1}^k \frac{\bar h_\epsilon(\hat k,k) s_{i[k]}}{\sum_{l = \hat k}^m \bar h_\epsilon(\hat k,l) s_{i[l]} q_{[l]}} g_\epsilon(\hat k)^+
\end{align}
where $\bar h_\epsilon$ and $g_\epsilon$ are defined as
\begin{align*}
\bar h_\epsilon(k,l) &:= h_\epsilon(q_{[k]} - q_{[l]}) \left[1 - \sum_{\ga = 1}^{k - 1} h_\epsilon(q_{[\ga]} - q_{[l]}) \prod_{\gb = 1}^{\ga - 1} \left[1 - h_\epsilon(q_{[\gb]} - q_{[l]})\right]\right]\\
g_\epsilon(k) &:= \begin{cases}\left[\bar p_i - \left(x_i + \sum_{j = 1}^n a_{ji} p_j\right)\right]^+ &\text{if } k = 0\\
    g_\epsilon(0) - \lambda^{\max}_i \left[\left(x_i + \sum_{l = 1}^m s_{il} q_l + \sum_{j = 1}^n a_{ji} p_j\right) - \bar p_i \right]^+ & \text{if } k = 1\\ 
    g_\epsilon(k-1) - \left[\frac{g_\epsilon(k-1)^+}{\sum_{l = k-1}^m \bar h_\epsilon(k-1,l) s_{i[l]} q_{[l]}} \wedge 1\right] \sum_{l = k-1}^m \bar h_\epsilon(k-1,l) s_{i[l]} q_{[l]} & \text{if } k \in \{2,...,m\}\end{cases}.
\end{align*}
One possible choice for the smoothing function $h_\epsilon$ is given by $h_\epsilon(z) = \exp(-\frac{1}{\epsilon}z)$.
This approximation strategy does a sequence of weighted proportional liquidations (see, e.g., Example~\ref{Ex:proportional-gamma}) with weights increasing with the price of assets.
\end{example}

\begin{example}\label{Ex:price-gamma-worst}
In contrast to Example~\ref{Ex:price-gamma-best}, a firm may wish to sell its lowest value assets first so that they will no longer be on its balance sheet.  Consider the same notation as above with $[k]$ denoting the index of the $k^\text{th}$ highest price.  The number of units of asset $[k]$ sold by firm $i$ is given by
\[\begin{split}&\gamma_{i[k]}(p,q) =\\
&\quad\frac{1}{q_{[k]}}\left(\left[\bar p_i - \left(x_i + \sum_{j = 1}^n a_{ji} p_j\right)\right]^+ - \lambda^{\max}_i \left[\left(x_i + \sum_{l = 1}^m s_{il} q_l + \sum_{j = 1}^n a_{ji} p_j\right) - \bar p_i\right]^+ - \sum_{l = k+1}^m s_{i[l]} q_{[l]}\right)^+.
\end{split}\]
As in Example \ref{Ex:price-gamma-best}, we may not be able to guarantee the existence of a solution under this liquidation strategy.  Define the net of functions $h_\epsilon: \bbr_+ \to [0,1]$ as above.  We can define an approximate liquidation strategy via the system of equations below.
\begin{align}
\label{Eq:price-gamma-worst}
\gamma_{i[k]}^\epsilon(p,q) &= \sum_{\hat k = k}^m \frac{\bar h_\epsilon(\hat k,k) s_{i[k]}}{\sum_{l = 1}^{\hat k} \bar h_\epsilon(\hat k,l) s_{i[l]} q_{[l]}} g_\epsilon(\hat k)^+\\
\nonumber \bar h_\epsilon(k,l) &:= h_\epsilon(q_{[l]} - q_{[k]}) \left[1 - \sum_{\ga = k+1}^m h_\epsilon(q_{[l]} - q_{[\ga]}) \prod_{\gb = \ga+1}^m \left[1 - h_\epsilon(q_{[l]} - q_{[\gb]})\right]\right]\\
\nonumber g_\epsilon(k) &:= \begin{cases}g_\epsilon(k+1) - \left[\frac{g_\epsilon(k+1)^+}{\sum_{l = 1}^{k+1} \bar h_\epsilon(k+1,l) s_{i[l]} q_{[l]}} \wedge 1\right] \sum_{l = 1}^{k+1} \bar h_\epsilon(k+1,l) s_{i[l]} q_{[l]} &\text{if } k \in \{1,...,m-1\}\\
    g_\epsilon(m+1) - \lambda^{\max}_i \left[\left(x_i + \sum_{l = 1}^m s_{il} q_l + \sum_{j = 1}^n a_{ji} p_j\right) - \bar p_i \right]^+ & \text{if } k = m\\
    \left[\bar p_i - \left(x_i + \sum_{j = 1}^n a_{ji} p_j\right)\right]^+ & \text{if } k = m+1\end{cases}. 
\end{align}
As with the approximate strategy in Example~\ref{Ex:price-gamma-best}, this liquidation strategy considers a sequence of weighted proportional liquidations with weights decreasing as the price of assets increase.
\end{example}

\begin{rem}
Examples~\ref{Ex:1stock-gamma} and \ref{Ex:proportional-gamma} present liquidation functions that are continuous and nonincreasing; Examples~\ref{Ex:price-gamma-best} and \ref{Ex:price-gamma-worst} present liquidation functions that are neither continuous nor nonincreasing.  The liquidation functions provided by Equations \eqref{Eq:price-gamma-best} and \eqref{Eq:price-gamma-worst} provide continuous approximations that converge to the respective examples as $\epsilon$ tends to $0$.
\end{rem}

\begin{rem}
Under the conditions of Theorem~\ref{Thm:clearing-exist}(ii), we can utilize the modified \emph{fictitious default algorithm} presented in \ci*{feinstein2015illiquid} to compute the greatest clearing payments and prices $(p^+,q^+)$.  Under the conditions of Theorem~\ref{Thm:clearing-exist}(i), we can search for the clearing payments and prices $(p^*,q^*)$ via fixed point iterations beginning from $(p^0,q^0) = (\bar p,\bar q)$.  However, the fixed point iterations may not converge as this is not a contraction mapping.  Alternative techniques for finding fixed points of continuous mappings, such as Scarf's algorithm (see, e.g., \ci*{scarf67}), can be employed instead.
\end{rem}

\section{Equilibrium liquidation strategies}\label{Sec:games}

Consider again the setting presented in Section~\ref{Sec:setting}.  We now wish to study the existence of the clearing payments and clearing prices when the liquidation strategies of each firm depend on the strategies chosen by all other firms.  That is, we wish to consider a game theoretic strategy.  And in particular we consider the case that all firms are value maximizers, i.e., firm $i$ wishes to maximize the valuation of its assets $s_i$.  Because the asset prices are influenced by the actions of all firms, this is an equilibrium liquidation strategy.  This is a more realistic scenario since financial institutions do work as value maximizers, and do not publicize their trading strategies beforehand.  Note that this is akin to the equilibrium liquidation strategy as presented in \ci*{feinstein2015illiquid}.

\begin{rem}
Though the problem is written as a function of the liquidations enacted by individual banks, the solution only depends on the aggregate liquidations of all other banks.  Notably, this solution does not require that a firm know who is selling, but only the total amount being sold within the financial system.
\end{rem}

Recall that the equity or loss of firm $i$ can be given by the form
\[x_i + \sum_{k = 1}^m s_{ik} q_k^* + \sum_{j = 1}^n a_{ji} p_j^* - \bar p_i\]  
for some clearing payment and price $(p^*,q^*)$.
Each firm, however, is given the choice on how to liquidate its own assets $s_i$, which can affect the prices of the illiquid assets $q^*$.  That is, every firm $i \in \{1,2,...,n\}$ can choose how many units $\gamma_{ik} \geq 0$ of asset $k \in \{1,2,...,m\}$ to liquidate so as to maximize its own equity.  In particular, the liquidation strategy for firm $i$ will depend on the liquidation strategies $\gamma_{-i}^*$ for all other firms.  In sum, firm $i$ will choose to liquidate in order to solve the maximization problem
\begin{align}\label{Eq:ind-opt}
\gamma_i(p,q,\gamma_{-i}^*) &\in \argmax_{g_i \in \Gamma_i(p,q)} s_i^\T F\left([s_i \wedge g_i] + \sum_{j \neq i} [s_j \wedge \gamma_j^*]\right)\\
\nonumber \Gamma_i(p,q) &= \left\{\gamma_i \in \bbr^m_+ \; | \; q^\T [s_i \wedge \gamma_i] = \left(q^\T s_i\right) \wedge \Lambda_i(p,q) \right\}\\
\nonumber \Lambda_i(p,q) &= \left(\left[\bar p_i - \left(x_i + \sum_{j = 1}^n a_{ji} p_j \right)\right]^+ - \lambda^{\max}_i \left[\left(x_i + \sum_{k = 1}^m s_{ik} q_k + \sum_{j = 1}^n a_{ji} p_j\right) - \bar p_i\right]^+\right)^+.
\end{align}
The formula provided by Equation~\eqref{Eq:ind-opt} identifies the number of units of each asset a firm should liquidate to maximize its own mark-to-market valuation.  Assumption~\ref{Ass:min-liquidation} restricts the allowable liquidations at pricing vector $q$, described by $\Gamma_i$ for each firm $i$.  Assuming all other firms follow the strategy $\gamma_{-i}^*$, then should firm $i$ act differently, it would be missing a potentially higher valuation.

In such a setting the clearing mechanism is modified to find the clearing payment $p^* \in [0,\bar p]$, clearing price $q^* \in [0,\bar q]$, and \emph{equilibrium liquidation strategy} $\gamma^* \in \bbr^{n \times m}_+$.  The modified clearing mechanism $\Psi: [0,\bar p] \times [0,\bar q] \times \bbr^{n \times m}_+ \to \pcal([0,\bar p] \times [0,\bar q] \times \bbr^{n \times m}_+)$ (where $\pcal$ denotes the power set) is defined by
\begin{align}\label{Eq:liquidation}
\begin{split}
\Psi(p,q,\gamma) := &\left\{\bar p \wedge (x + Sq + A^\T p)\right\} \times \left\{F\left(\sum_{i = 1}^n [s_i \wedge \gamma_i]\right)\right\} \\ 
& \quad \times \prod_{i = 1}^n \argmax_{g_i \in \Gamma_i(p,q)} s_i^\T F\left([s_i \wedge g_i] + \sum_{j \neq i} [s_j \wedge \gamma_j]\right).
\end{split}
\end{align}
A fixed point of the clearing mechanism, $(p^*,q^*,\gamma^*) \in \Psi(p^*,q^*,\gamma^*)$, is the joint clearing payment, price, and liquidation strategies.
Note that each firm satisfies the minimum liquidation condition described in Assumption~\ref{Ass:min-liquidation} because $q^* = F(\sum_{i = 1}^n [s_i \wedge \gamma_i^*])$ by definition.  Also note that no firm can unilaterally increase its own valuation by trading in a contrary manner to that given by $\gamma^*$.  In Theorem~\ref{Thm:equil-liquidation}, we give conditions for the existence of the joint clearing payments, prices, and liquidation strategies under the modified clearing mechanism $\Psi$.

\begin{theorem}\label{Thm:equil-liquidation}
Consider a financial system $(A,\bar{p})$ with liquid endowments $x \in \bbr^n_+$ and illiquid endowments $S \in \bbr^{n \times m}_+$.  Consider some inverse demand function $F$ satisfying Assumption~\ref{Ass:idf} such that $F(\sum_{i = 1}^n s_i) \in \bbr^m_{++}$ and $\beta \mapsto s_i^\T F(\beta)$ is quasi-concave for all $i = 1,2,...,n$.  There exists a combined clearing payment, clearing price, and equilibrium liquidation strategy, i.e., there exists $(p^*,q^*,\gamma^*) \in \Psi(p^*,q^*,\gamma^*)$.
\end{theorem}
\begin{proof}
This follows identically to Theorem~4.1 of \ci*{feinstein2015illiquid}.
\end{proof}

Financially, if the inverse demand function is concave -- combined with Assumption~\ref{Ass:idf} -- then the velocity of price impact increases as more assets are liquidated.  In financial markets this would be the case if the limit order book is dense near the initial market price but with a long shallow tail at lower prices.  That is, the underlying liquidity in the market dries up as more assets are sold leading to significant asset price movements.

\begin{rem}
The clearing liquidation strategy $\gamma^*$ is a Nash equilibrium for the financial system if all firms are value maximizers.  That is, no firm $i$ has an incentive to change strategies from $\gamma_i^*$ since it is a maximizer of the valuation problem (given the full market strategy $\gamma^*$).
\end{rem}

\begin{rem}
Under the conditions of Theorem~\ref{Thm:equil-liquidation} and if $\beta \mapsto s_i^\T F(\beta)$ is strictly quasi-concave, then $\Psi$ is a singleton.  Thus, if there exists a limit to some fixed point iterations of $\Psi$ then this limit is a joint clearing payment, price, and liquidation strategy.  For computation we will run fixed point iterations beginning from $(p^0,q^0,\gamma^0) = (\bar p,\bar q,0)$.  In the case that no fixed point could be found due to non-convergence of the fixed point iterations we will interpolate linearly between clearing solutions for known network parameters; this effect will be seen in Example~\ref{Ex:multiasset} in the next section.
\end{rem}

\section{Numerical case studies}\label{Sec:cases}

In this section we implement the proposed leveraged financial contagion model.  The underlying asset and liability data is a subsection of US banking data discussed in Appendix~\ref{Sec:data}.  With this data we have choices for the structure of the financial network which impacts the risk levels:
\begin{itemize}
\item the value of assets that are from interbank liabilities of other firms ($\bar p_{\cdot i}$ for firm $i$);
\item the number of illiquid assets ($m$) and the composition of the initial portfolio of liquid ($x_i$ for firm $i$) and illiquid ($s_i$ for firm $i$) holdings;
\item the liquidation strategy ($\gamma_i$ for firm $i$) of illiquid assets;
\item the network structure and distribution of liabilities ($\bar p_{i \cdot}$ for firm $i$) within and to outside the system; and
\item the leverage requirement ratio.
\end{itemize}

To construct the financial network it is sometimes beneficial to consider the case that a fraction of the liabilities of the financial firms leave the system or come in from outside of the financial system.  To accommodate this we consider an augmented system with an additional ``firm'', denoted by node $0$.  Without loss of generality, as discussed in \ci*{feinstein2014measures}, we assume the outside node $0$ will never default on its obligations -- and more specifically we will construct a system in which node $0$ has no obligations to the rest of system.  When desired, in order to incorporate the possibility of node $0$ not paying its obligations in full, we will stress the assets for the different firms at the start of the clearing mechanism.

To calibrate the network we first consider some (possibly random) structure to the directed links.  Notationally, let a link from $i$ to $j$ have (possibly) positive weight if $j \in I(i)$.  With a given maximum percentage of the assets for each firm coming from incoming liabilities -- with nominal size of incoming obligations $P_i^\text{in}$ -- we run the linear program~\eqref{Eq:lin-prog_network} to find the weights for each arc in the network.  The solution to this optimization problem is the matrix of obligations (including the outside node $0$) so that the minimal amount of liabilities is passed out of the initial financial system and the network structure is retained.
\begin{equation}\label{Eq:lin-prog_network}
\min_{L \in \bbr^{(n+1) \times (n+1)}_+} \sum_{j = 1}^n L_{j0} \quad \text{s.t.} \quad \sum_{j = 1}^n L_{ji} \leq P_i^\text{in}, \; \sum_{j = 1}^n L_{ij} = \bar p_i, \; L_{ii},L_{0i},L_{ij_i} = 0 \quad \forall i,\forall j_i \in I(i)^c
\end{equation}
Though we only present this calibration mechanism, the qualitative results presented herein appear to hold generally.  That is, the results appear robust to network topology; this is in line with the results of \ci*{GY14}. 

Throughout the examples we will consider three proxies of systemic risk where appropriate:
\begin{itemize}
\item The fraction of total firms that are defaulting on a portion of their obligations.  This was studied previously in, e.g., \ci*{lehar05} and \ci*{zhou10}.
\item The fraction of total firms that are violating the maximum leverage ratio requirement.  This can equivalently be given (except on space of Lesbegue measure 0) by the firms that are liquidating all of their assets.  Note that this value, in theory, can be independent of a firm defaulting on its obligations.
\item The fraction of payments to the economy outside of the financial network to the total amount owed.  Mathematically this can be given by the formula
\[\frac{\sum_{j = 1}^n a_{j0} p_j}{\sum_{j = 1}^n a_{j0} \bar p_j}.\]
This was studied previously in \ci*{feinstein2014measures} for use as an aggregation of the health of the financial network in systemic risk measures. 
\end{itemize}

In the construction of the portfolio holdings, we consider two measures of diversification.  First the percentage of initial portfolio holdings in the liquid asset (as opposed to invested in the illiquid assets more generally).  Second is a value of how much is invested in each illiquid asset. 
Consider the case in which firm $i$ has initial capital $C_i$ in its portfolio, invests $\alpha \in [0,1]$ fraction in the liquid asset, and assume random distribution of investments in the illiquid assets defined by some standard deviation $\sigma \geq 0$.  The value of $\sigma$ gives a parameter for the diversification.  Under this scheme, firm $i$ will hold $x_i = \alpha C_i$ in the liquid asset.  To construct the illiquid holdings, let
\[c_i^k \sim N(1,\sigma^2)^+ \quad (\forall k = 1,2,...,m)\]
where $N^+$ denotes the maximum of a normal distribution and $0$, then 
\[s_{ik} = \begin{cases}(1-\alpha)\frac{c_i^k}{\sum_{l = 1}^m \bar q_l c_i^l} C_i & \text{if } \sum_{l = 1}^m c_i^l > 0\\ \frac{1-\alpha}{\sum_{l = 1}^m \bar q_l}C_i &\text{else}\end{cases}.\]
under the assumption that $F(0) = \bar q \in \bbr^m_{++}$.  Note that when $\sigma = 0$ this corresponds to each firm dividing its initial illiquid holdings equally amongst the assets, i.e., perfect diversification.  As $\sigma$ increases, the disparity between assets increases and (probabilistically) each firm increases the number of assets in which it is not invested.  Further note that if $\alpha = 1$, this scheme reduces to the contagion model presented in \ci*{EN01} as no internal control remains for lowering a firm's leverage ratio.

We will consider a few examples with the US banking data.  First, in Example~\ref{Ex:1asset-portfolio}, we study the case in which there is only a single, representative, illiquid asset and vary the portfolio composition between the liquid and illiquid asset.  In Example~\ref{Ex:1asset-network}, we continue to study the case with a single illiquid asset, but now vary the composition of the network -- varying the probability of a connection being formed and varying the value of the assets from interbank liabilities.  In Example~\ref{Ex:multiasset}, we study the implications of different diversification (defined by the parameter $\sigma$) and liquidation strategies have on the risk of the system by incorporating multiple illiquid assets. Finally, in Example~\ref{Ex:leverage}, we consider a counterfactual scenario in which we shock the liquid and illiquid holding values after allowing for firms to cancel some assets and liabilities to meet the leverage requirement but before the clearing mechanism is considered.

\begin{example}\label{Ex:1asset-portfolio}
Consider a market with a single (representative) illiquid asset as well as the liquid asset.  Assume a network with random connections; the probability of a link from node $i$ to $j$ (for $i \neq j$) is $25\%$ and independently sampled for each potential link.  Further, let at most $10\%$ of any firm's assets be from incoming liabilities from other firms.  As there is only a single asset the liquidation strategy is uniquely defined in Example~\ref{Ex:1stock-gamma}.  Let the inverse demand function be given by
\[F(z) = \begin{cases} 1 - \frac{2z}{3 \times 10^8}& \text{if } z \leq 5 \times 10^7 \\ \frac{10^4\sqrt{2}}{3\sqrt{z}}& \text{if } z \geq 5 \times 10^7 \end{cases}.\]

When the entire portfolio is in the liquid asset ($\alpha = 1$) then under this setup the network is always stable -- no defaults, no leverage shortfalls, all payments made in full as anticipated by the results of \ci*{GY14}.  If however the entire portfolio is in the illiquid asset ($\alpha = 0$) then there is a minimal leverage requirement for which any lower maximum leverage leads to all firms defaulting and failing the leverage requirement and any higher maximum leverage leads to no firms defaulting on any obligations.  This is evidenced in, e.g., Figure~\ref{Fig:1asset_portfolio_2007}; note that the fraction of firms defaulting and the fraction that violate the leverage requirement coincide exactly.  In Table~\ref{Table:1asset_portfolio_switch} the values of this minimal acceptable leverage requirement for each year are given.  From this table it is clear that there was a crisis that occurred in 2008; the minimum possible leverage requirement so that the system is stable is $19.8325$ which is much higher than the values in every other year.  In the years after the 2008 financial crisis the system became more stable until 2012 when the minimal safe leverage ratio started to slowly rise again.  This would be an interesting metric to follow going forward as it is a simple measure and is reflected in all three proxies for systemic stability.
\begin{table} \label{Table:1asset_portfolio_switch}
\centering
\begin{tabular}{|c|*{8}{c}|}
\hline
Year & 2007 & 2008 & 2009 & 2010 & 2011 & 2012 & 2013 & 2014 \\ \hline
Leverage & 14.6325 & 19.8325 & 14.2700 & 12.9125 & 12.0375 & 11.8250 & 12.6700 & 13.0850\\ \hline
\end{tabular}
\caption{Example~\ref{Ex:1asset-portfolio}: Minimal acceptable leverage requirements from 2007-2014.}
\end{table}

\begin{figure}
\centering
\includegraphics[height=0.4\textheight,width=0.8\textwidth]{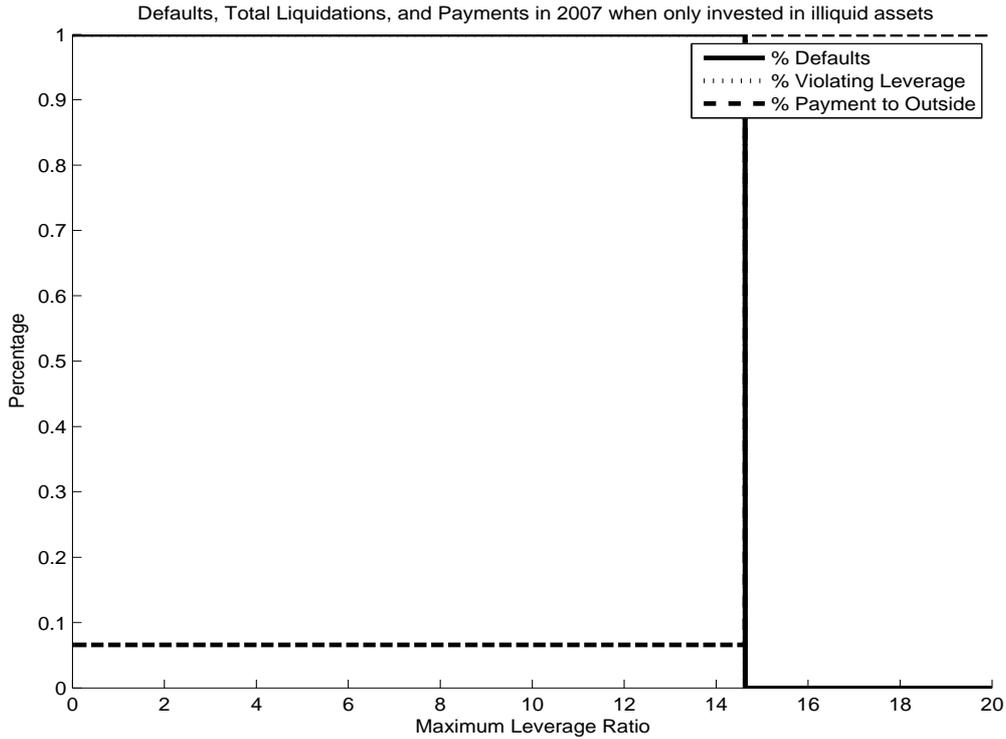}
\caption{Example~\ref{Ex:1asset-portfolio} for 2007: Fraction of firms that default on a portion of their debt, fraction of firms that have a terminal leverage shortfall, and fraction of liabilities paid to outside the financial system when all portfolio investments are made in the illiquid asset.}
\label{Fig:1asset_portfolio_2007}
\end{figure}

Figure~\ref{Fig:1asset_alpha_2007} shows the general shape of what happens as the fraction of portfolio holdings $\alpha$ in the liquid asset (as opposed to the illiquid asset) increases from $0$ to $1$.  The result here is exactly as anticipated.  The larger the holdings in the cash asset the more resilient the financial system since fire sales and price fluctuations have a lower impact on a firm's overall valuation and equity.  Figure~\ref{Fig:1asset_alpha_minimum_2007} provides a plot of the minimal acceptable leverage requirements in 2007 as a function of the fraction of portfolio holdings $\alpha$ in the liquid asset.  Notably, the more assets are held in the liquid asset, the more resilient the system becomes to forced liquidations to satisfy a leverage requirement.  This curve is nearly linear from $\alpha = 0\%$ until nearly $\alpha = 80\%$, which is to be expected as the leverage is a function of the illiquid assets.  

\begin{figure}
\centering
\includegraphics[height=0.4\textheight,width=0.8\textwidth]{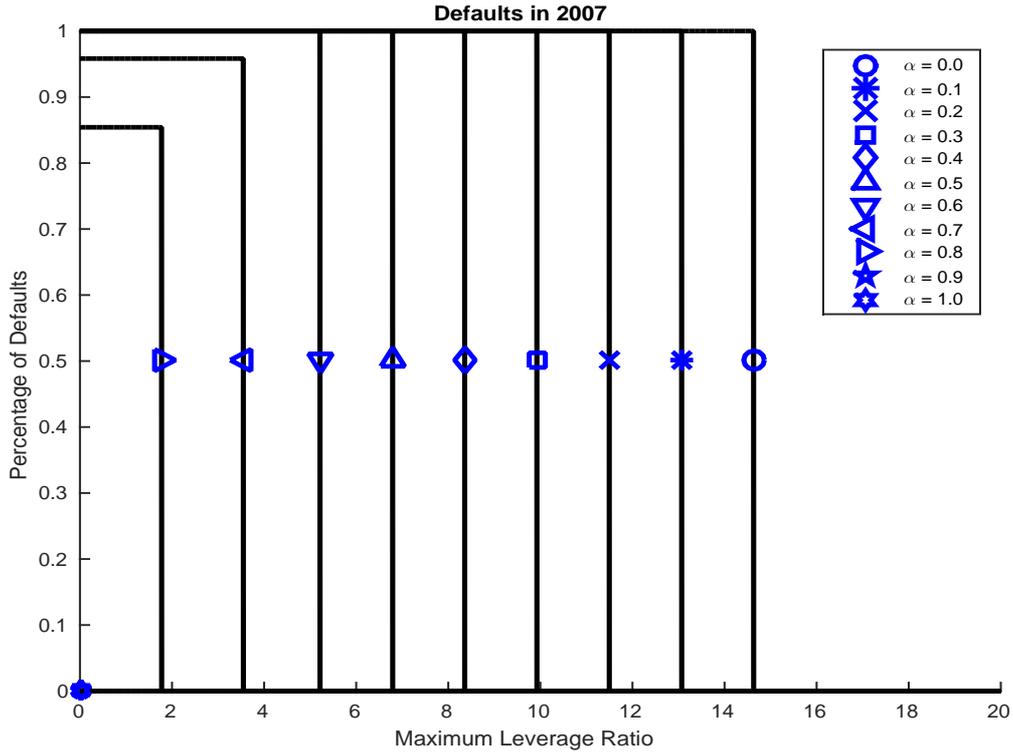}
\caption{Example~\ref{Ex:1asset-portfolio} for 2007: Fraction of firms that default on a portion of their debt over varying diversification between the liquid and illiquid asset ($\alpha$).}
\label{Fig:1asset_alpha_2007}
\end{figure}

\begin{figure}
\centering
\includegraphics[height=0.4\textheight,width=0.8\textwidth]{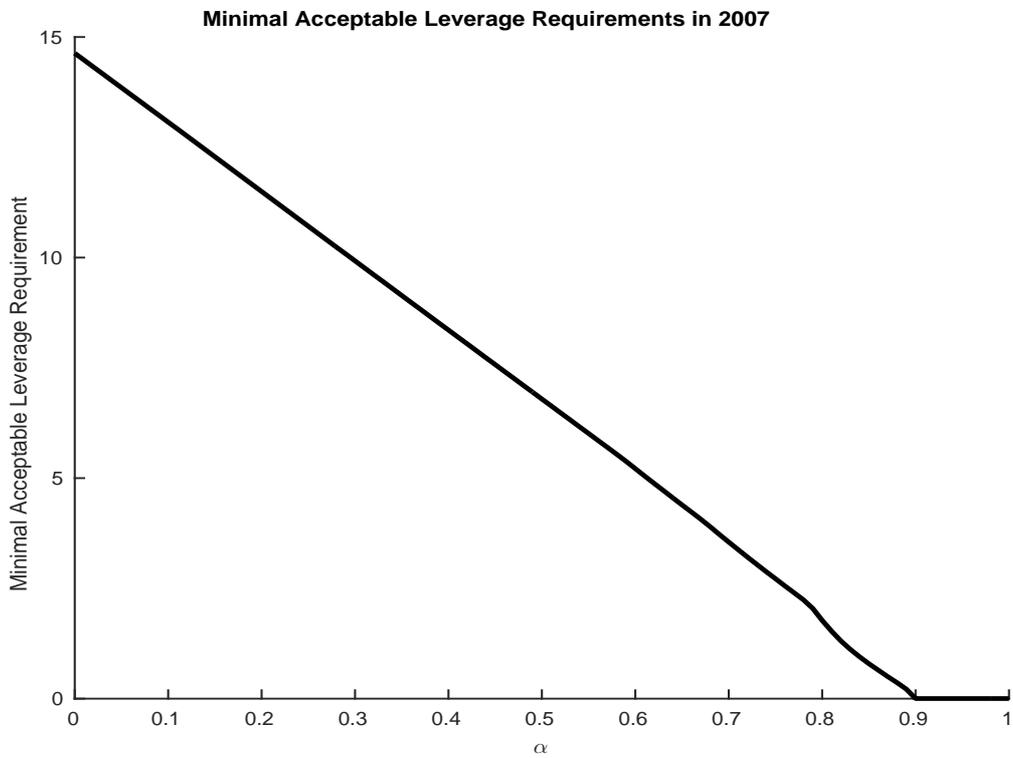}
\caption{Example~\ref{Ex:1asset-portfolio} for 2007: Minimal acceptable leverage requirements over varying diversification between liquid and illiquid asset ($\alpha$).}
\label{Fig:1asset_alpha_minimum_2007}
\end{figure}
\end{example}

\begin{example}\label{Ex:1asset-network}
Again let us consider a market with a single (representative) illiquid asset as well as the liquid asset.  Assume a network with random connections, though in this example we will vary the probability of a connection being formed.  Additionally, we will vary the maximum percentage of a firm's assets which come from interbank liabilities. Throughout this example we will fix the portfolio holdings so that $\alpha = 20\%$ of the holding's initial value is in the liquid asset. As there is only a single asset the liquidation strategy is uniquely defined in Example~\ref{Ex:1stock-gamma}.  And we will consider the same inverse demand function as in the prior example, i.e., 
\[F(z) = \begin{cases} 1 - \frac{2z}{3 \times 10^8}& \text{if } z \leq 5 \times 10^7 \\ \frac{10^4\sqrt{2}}{3\sqrt{z}}& \text{if } z \geq 5 \times 10^7 \end{cases}.\]

In Figure~\ref{Fig:1asset_network_pin_2007} we plot the fraction of obligations to the outside node $0$ that are actualized after the clearing mechanism is completed and in Figure~\ref{Fig:1asset_network_minimum_2007} we plot both the minimal acceptable leverage requirements and the fraction of obligations to the outside node $0$ that are actualized after the clearing mechanism is completed when the leverage requirements are set to 0 with the data from 2007.  As with Example~\ref{Ex:1asset-portfolio} above, we choose to display only the plots from 2007 as it is a representative result for each subsequent year as well.  In fact, in Figure~\ref{Fig:1asset_network_pin_2007}, we plot only the fraction of obligations that are realized because the plot for fraction of firms which defaults and the fraction of firms which liquidate all their assets have the same structure with critical leverage levels displayed in Figure~\ref{Fig:1asset_network_minimum_2007}.  In this scenario we assume the probability of forming a connection between any two distinct firms is fixed at $40\%$.  Throughout we vary the maximal allowed fraction of initial assets that come from the interbank liabilities.  
Most interestingly about these plots is that, though as the fraction of assets from interbank liabilities increases the system is more stable (a higher fraction of liabilities are paid) at low leverage requirements, there is a point where increasing the percentage of interbank liabilities makes the system less robust in view of the critical leverage level.  It is to be expected that at low leverage requirements greater connectivity equates to more robust system behavior -- as a larger fraction of assets are from interbank liabilities, the banks become less susceptible to fire sales.  However, interestingly, there is a connectivity level beyond which the contagion through direct linkages overwhelms the beneficial aspects of risk sharing.
\begin{figure}
\centering
\includegraphics[height=0.4\textheight,width=0.8\textwidth]{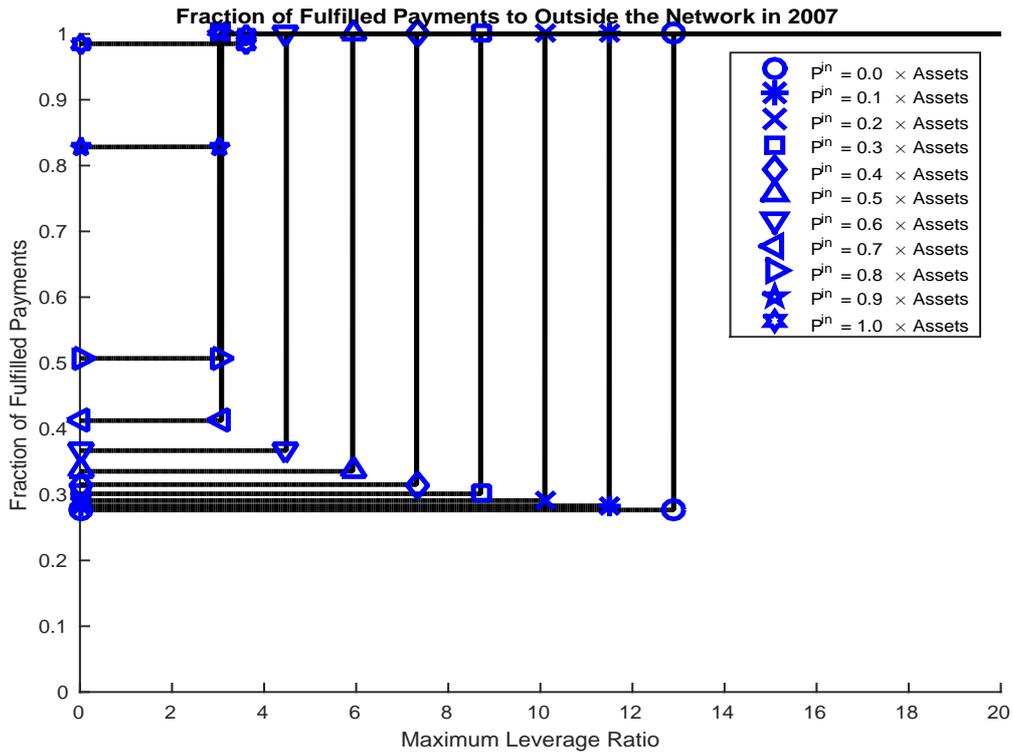}
\caption{Example~\ref{Ex:1asset-network} for 2007: Fraction of obligations to outside the system that are realized varying the percentage of assets that come from interbank liabilities ($P^\text{in}$) when the probability of forming a connection within the system is $40\%$.}
\label{Fig:1asset_network_pin_2007}
\end{figure}

\begin{figure}
\centering
\includegraphics[height=0.4\textheight,width=0.8\textwidth]{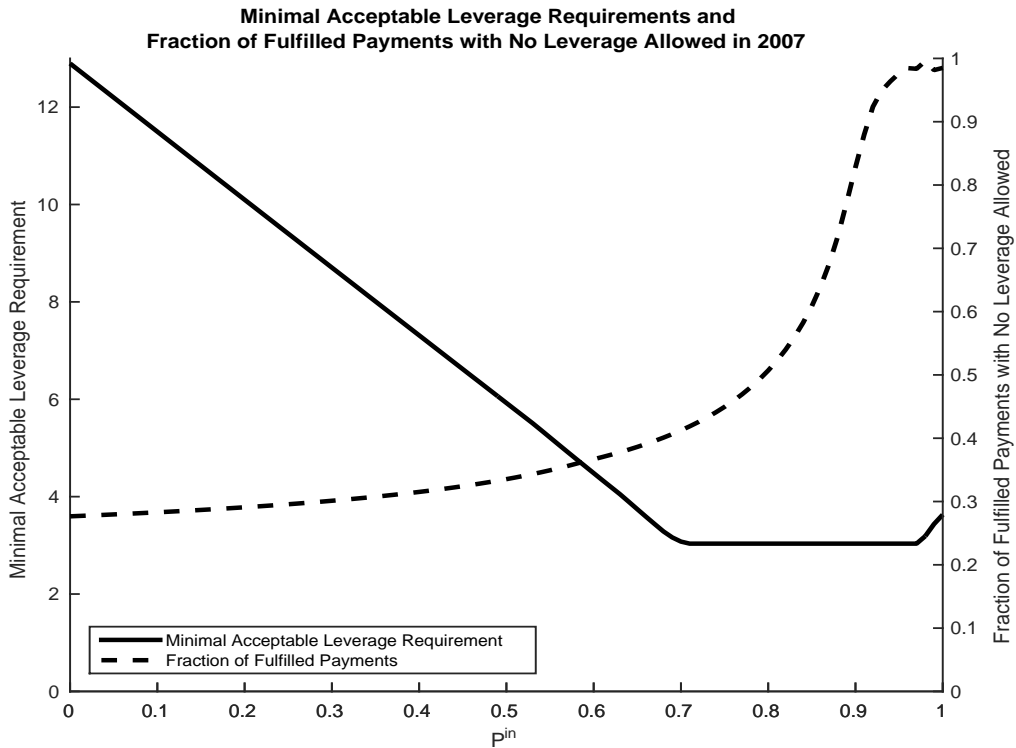}
\caption{Example~\ref{Ex:1asset-network} for 2007: Minimal acceptable leverage requirements [left axis] and fraction of fulfiilled payments to outside the financial system when leverage requirement is set at 0 [right axis] varying the percentage of assets that come from interbank liabilities ($P^\text{in}$) when the probability of forming a connection within the system is $40\%$.}
\label{Fig:1asset_network_minimum_2007}
\end{figure}

In Table~\ref{Table:1asset_network_prob} we provide the minimal safe leverage requirements in each year given varying connection probabilities.  For this situation we fix the maximal fraction of a firms assets from incoming liabilities at $10\%$.  Interestingly, though a clear trend, the minimal acceptable leverage requirements are not monotonically decreasing with the connection probabilities.  This is even when normalizing for the realized sampled number of connections.  As opposed to Figure~\ref{Fig:1asset_network_pin_2007}, since there is a maximum percent of a firms assets in interbank liabilities, the increased connectivity tends to improve the systemic risk via risk sharing properties.  Finally, as observed in Example~\ref{Ex:1asset-portfolio} previously and as expected, the financial system appears to be in its worst shape (by a significant margin) in 2008 and appears to be the most stable in 2012 before trending more risky after that trough.
\begin{table} \label{Table:1asset_network_prob}
\centering
\begin{tabular}{|c|*{8}{c}|}
\hline
\backslashbox{Prob.}{Year} & 2007 & 2008 & 2009 & 2010 & 2011 & 2012 & 2013 & 2014 \\ \hline
0.00 & 12.8975 & 17.5200 & 12.5750 & 11.3675 & 10.5875 & 10.4000 & 11.1525 & 11.5200 \\
0.01 & 12.8350 & 17.5200 & 12.5750 & 10.1025 & 10.1450 & ~9.7875 & 11.1425 & 11.5200 \\
0.02 & 12.8975 & 17.5200 & 11.2175 & 10.1025 & ~9.4300 & 10.2700 & 10.8200 & 10.2675 \\
0.03 & 11.4975 & 17.0800 & 11.2175 & 11.0200 & 10.0675 & ~9.7550 & 11.1525 & 11.5200 \\
0.04 & 11.4975 & 15.6675 & 11.2175 & 10.1025 & ~9.7600 & ~9.7550 & ~9.9375 & 10.2675 \\
0.05 & 11.4975 & 15.6675 & 11.2175 & 10.8550 & ~9.4300 & ~9.2600 & ~9.9375 & 10.5275 \\
0.06 & 12.8975 & 15.6675 & 11.2175 & 10.1025 & ~9.4300 & ~9.2600 & ~9.9375 & 10.2675 \\
0.07 & 11.4975 & 15.6675 & 11.2175 & 10.1025 & ~9.4300 & ~9.2600 & 10.5625 & 10.2675 \\
0.08 & 11.4975 & 15.6675 & 12.5750 & 10.1025 & ~9.5525 & ~9.2600 & ~9.9375 & 10.2675 \\
0.09 & 11.4975 & 15.6675 & 11.2175 & 10.1025 & ~9.4300 & ~9.2600 & ~9.9375 & 10.2675 \\
0.10 & 11.4975 & 15.6675 & 11.2175 & 10.1025 & 10.0675 & ~9.2600 & ~9.9375 & 10.2675 \\
0.20 & 11.4975 & 15.6675 & 11.2175 & 10.1025 & ~9.4300 & ~9.2600 & ~9.9375 & 10.2675 \\
0.30 & 11.4975 & 15.6675 & 11.2175 & 10.1025 & ~9.4300 & ~9.2600 & ~9.9375 & 10.2675 \\
0.40 & 11.4975 & 15.6675 & 11.2175 & 10.1025 & ~9.4300 & ~9.2600 & ~9.9375 & 10.2675 \\
0.50 & 11.4975 & 15.6675 & 11.2175 & 10.1025 & ~9.4300 & ~9.2600 & ~9.9375 & 10.2675 \\ \hline
\end{tabular}
\caption{Example~\ref{Ex:1asset-network}: Minimal acceptable leverage requirements as a function of the probability of a connection.}
\end{table}
\end{example}

\begin{example}\label{Ex:multiasset}
Let us now consider a market with $m = 10$ illiquid assets.  As before assume that $\alpha = 20\%$ of the portfolio's initial wealth is in the liquid asset.  The remainder of the portfolio holdings will be spread amongst the illiquid assets with varying parameter $\sigma$.  The inverse demand function for every asset $k$ is given by the concave inverse demand function 
\[F_k(z) = 1 - \left(\frac{z_k}{(2.5)^{2/3}\sum_{i = 1}^n s_{ik} +10^{-10}}\right)^{3/2}.\]  
Note that, for a fixed system-wide portfolio $S$, on the domain of interest (i.e., $[0,\sum_{i = 1}^n s_i]$) the inverse demand function $F$ is nonincreasing and strictly concave.
The asset liquidations will be done among varying strategies presented previously: proportional liquidation in Example~\ref{Ex:proportional-gamma}, (a continuous approximation of) liquidating the worst performing assets first in Example~\ref{Ex:price-gamma-worst}, and the equilibrium liquidation strategy in Section~\ref{Sec:games}.  Further assume the network with $25\%$ chance of a connection between any two distinct firms and at most $10\%$ of any firms initial assets come from interbank liabilities.

The first set of comparisons was taken over different liquidation strategies at diversification $\sigma = 0.5$.  In particular, we compute the system for proportional liquidation (Example~\ref{Ex:proportional-gamma}), a continuous approximator of worst asset liquidation (Example~\ref{Ex:price-gamma-worst}), and the equilibrium liquidation strategy (Section~\ref{Sec:games}).  As was noticed in Example 5.4 of~\ci*{feinstein2015illiquid}, the liquidation strategy does not have a strong impact on the risk level as measured by the size of defaults.  
Though Figure~\ref{Fig:multiasset_2007} appears to show a large deviation for the equilibrium trading strategy, this is a result of no convergent clearing payment (and linear interpolation) along the sloped part of the curve.  Noticeably the proportional liquidation is safer than the worst asset liquidation strategy.  This is as expected due to the implications of concavity of the inverse demand function on the velocity of price impact.
\begin{figure}
\centering
\includegraphics[height=0.4\textheight,width=0.8\textwidth]{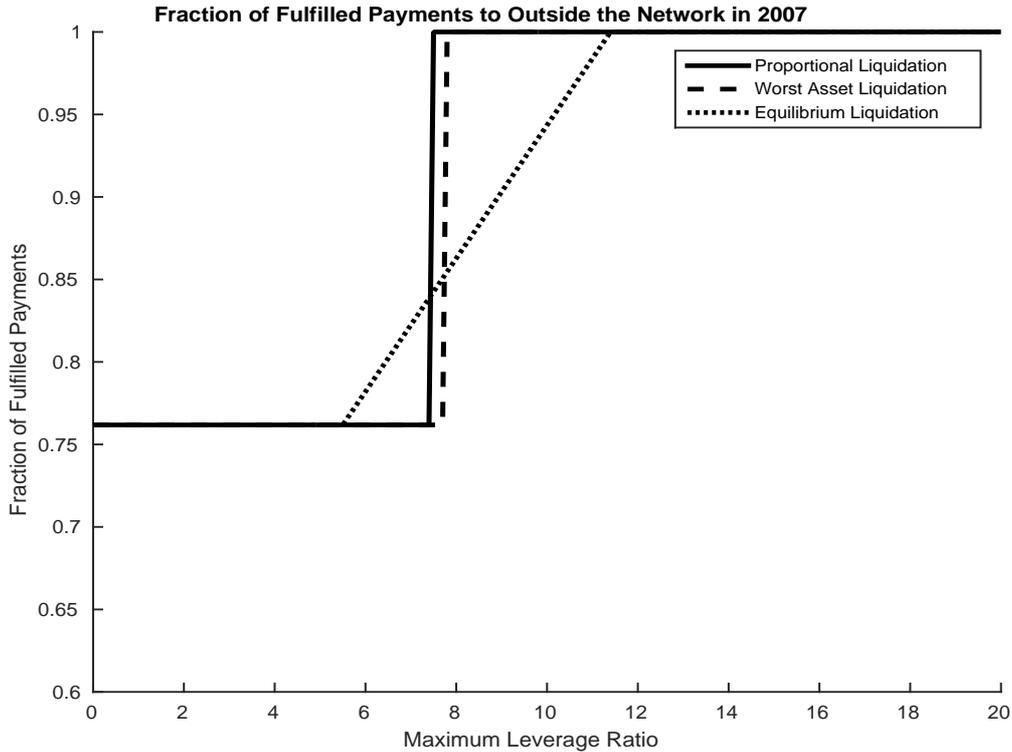}
\caption{Example~\ref{Ex:multiasset} for 2007: Comparing risk levels under varying liquidation strategies with diversification parameter $\sigma = 0.5$.}
\label{Fig:multiasset_2007}
\end{figure}

Since all liquidation strategies perform comparably, we will now consider only the proportional liquidation strategy.  Here we consider varying diversification.  Interestingly, the minimal safe leverage requirement is only very weakly dependent on the diversification parameter $\sigma$ with a high probability.  
In Figure~\ref{Fig:multiasset_minimum_confidence_2011} we plot the minimal acceptable leverage requirements as a function of the diversification from 2011 over 1000 simulated portfolios.  We chose this year as it is representative, but shows a more distinct difference than the other years; we skip the plot of the number of firms defaulting as the critical leverage requirement is a threshold that below it 100\% of banks are defaulting and above 0\% are.  In general, though not monotonic in every realization, greater diversification (i.e., lower $\sigma$) leads to a safer system as would be expected.  As demonstrated by the confidence intervals in Figure~\ref{Fig:multiasset_minimum_confidence_2011}, the behavior of the system with respect to the diversification is very sensitive to the actualized parameters; this demonstrates the truly nonlinear nature between fire sales and the critical leverage threshold.
\begin{figure}
\centering
\includegraphics[height=0.4\textheight,width=0.8\textwidth]{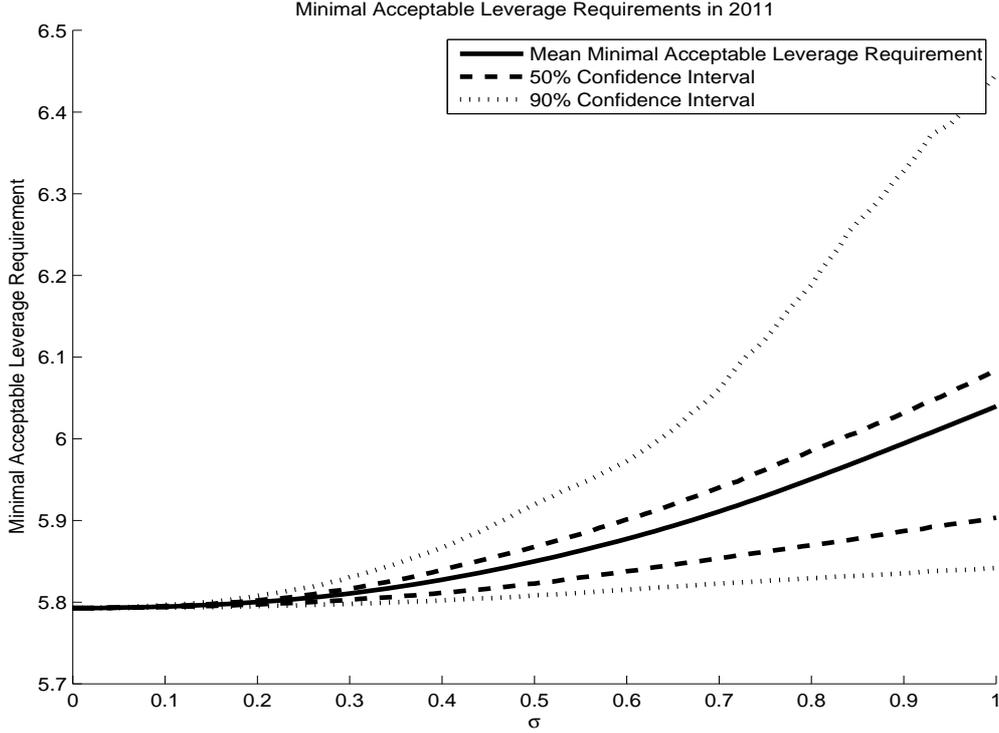}
\caption{Example~\ref{Ex:multiasset} for 2011: Minimal acceptable leverage requirements under proportional liquidation with varying diversification parameters ($\sigma$).} 
\label{Fig:multiasset_minimum_confidence_2011}
\end{figure}
\end{example}

\begin{example}\label{Ex:leverage}
Finally, we will consider a counterfactual scenario of assets and liabilities generated from the US banking data discussed in Appendix~\ref{Sec:data}.  In this example we will allow firms to borrow or lend money to increase both liabilities and assets by the same value $w_i$ so that they begin the system at the leverage requirement, i.e.,
\[\frac{\bar p_i + w_i}{x_i + \sum_{k = 1}^m s_{ik} q_k + \sum_{j = 1}^n a_{ji} \bar p_j - \bar p_i} = \lambda^{\max}_i.\]
Additionally, in order to prompt a systemic crisis we assume that, before the clearing mechanism is run, all firms have their portfolio holdings reduced by some parameter $\beta \in [0,1]$ due to a systematic risk factor.  That is, if initial portfolio holdings in the liquid and illiquid assets were $x_i$ and $s_i$ (respectively) for some firm $i$, after the stress the same holdings are given by $(1-\beta)x_i$ and $(1-\beta)s_i$.  Note that though in this scenario we are stressing the number of assets held, the same effect could be accomplished by stressing the initial prices of the assets.

For the network setup, consider a market with a single (representative) illiquid asset with inverse demand function
\[F(z) = \begin{cases} 1 - \frac{2z}{3 \times 10^8}& \text{if } z \leq 5 \times 10^7 \\ \frac{10^4\sqrt{2}}{3\sqrt{z}}& \text{if } z \geq 5 \times 10^7 \end{cases}.\]
As before, let the initial portfolio for each firm will be $\alpha = 20\%$ in the liquid asset and $80\%$ invested in the illiquid asset.  Since there is only a single illiquid asset, the liquidation strategy is uniquely defined by Example~\ref{Ex:1stock-gamma}.  To construct the network assume a probability of connections to be $25\%$ and the maximum of $10\%$ of any firms assets come from interbank liabilities.

First we study what happens to the number of defaulting firms, the number of firms that violate the leverage requirement after the clearing mechanism completes, and the amount of obligations paid to the outside node $0$.  This is shown in Figure~\ref{Fig:leverage_2007} for the data of year 2007; each subsequent year is of the same form.  Note that the fraction of firms defaulting and the fraction that violate the leverage requirement coincide exactly.  Of particular note is that the system is stable until some threshold leverage requirement as in the previous examples.  But now, the direction of increased risk is flipped from the prior examples, i.e., lower leverage requirements result in more stable systems during a systematic shock.  The reason for the difference between this result and the prior ones is that previously the initial leverage was not dependent on the leverage requirements.  Therefore a lower leverage requirement leads to greater liquidations and fire sales.  However, in this scenario lower leverage requirements lead to fewer investments and thus lower sensitivity to price fluctuations.  The second note for this plot is that the curve for the payments to the outside economy is no longer flat after the threshold leverage requirement.  That is, when defaults do occur, the size of losses increases the greater the allowable leverage becomes.
\begin{figure}
\centering
\includegraphics[height=0.4\textheight,width=0.8\textwidth]{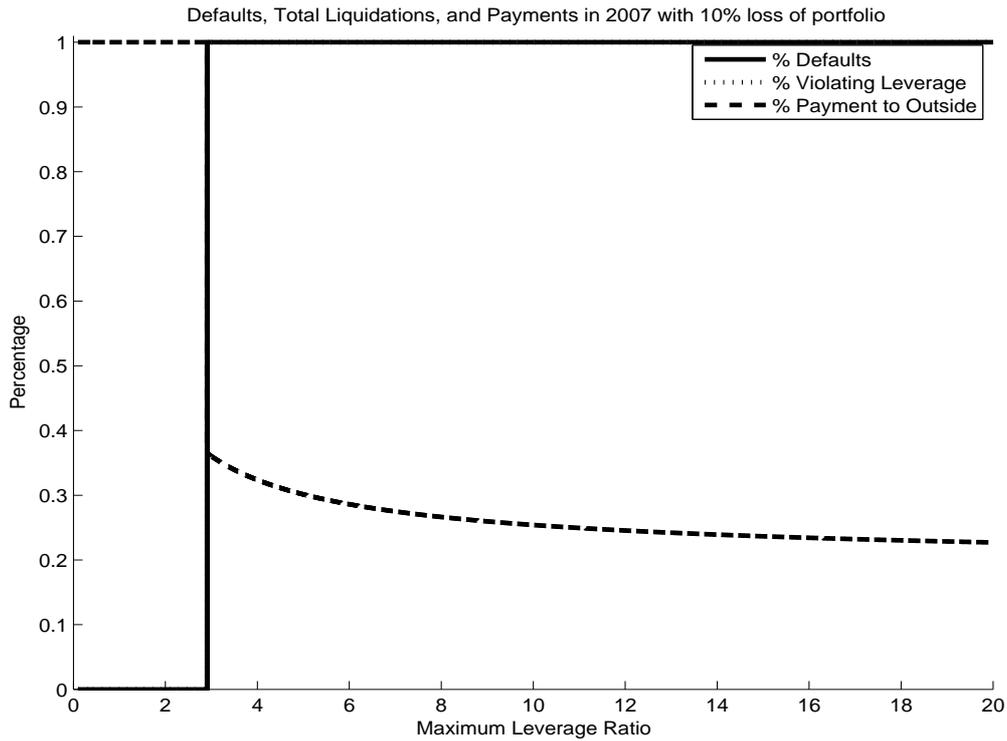}
\caption{Example~\ref{Ex:leverage} for 2007: Fractional number of firms defaulting, liquidating all assets, and amount of obligations paid to outside the system for a fixed shock $\beta = 10\%$.}
\label{Fig:leverage_2007}
\end{figure}

In Figure~\ref{Fig:leverage_beta_2007} we study 2007 yet again as a representative of all following years.  We focus on the fraction of liabilities being paid to the outside node as a function of the shock size $\beta$.  As expected, the larger the shock size, the larger the losses to the outside node and the earlier the losses occur.  In particular, without a systematic shock there are no losses within the system.  But even at $\beta = 10\%$ valuation loss, the systemic risk becomes significant even for low leverage ratios.  As the shock size increases, the percentage of losses to the outside node grow, as is expected.
\begin{figure}
\centering
\includegraphics[height=0.4\textheight,width=0.8\textwidth]{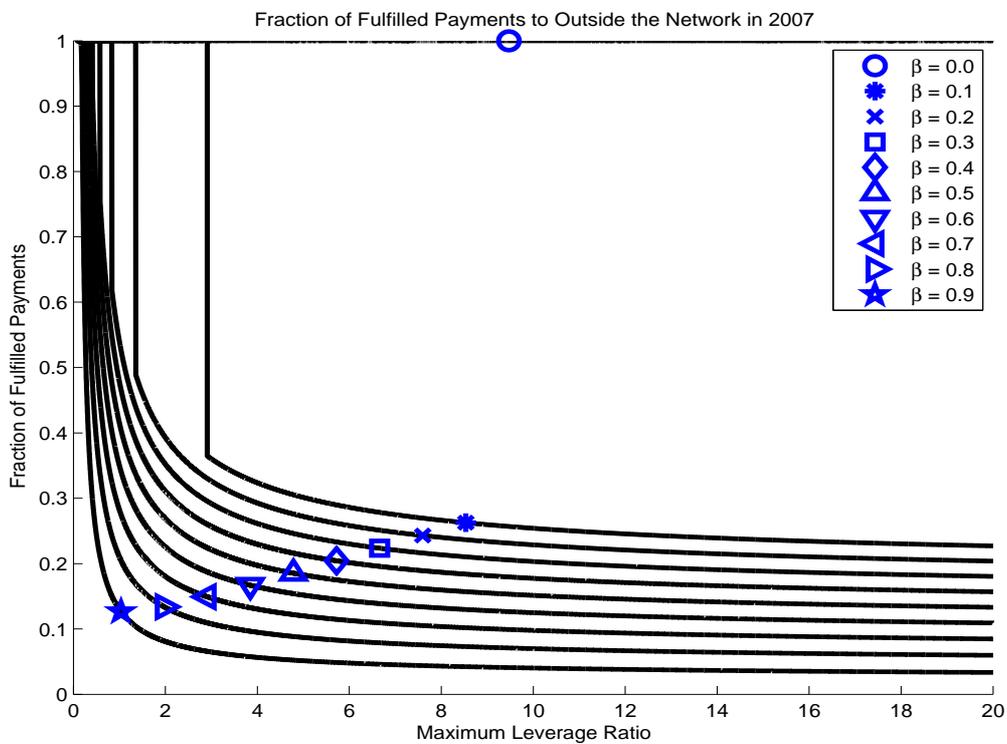}
\caption{Example~\ref{Ex:leverage} for 2007: Fraction of obligations to outside the system that are realized varying the shock size ($\beta$).}
\label{Fig:leverage_beta_2007}
\end{figure}
\end{example}

\section{Conclusion}
In this paper we extended the financial contagion model of \ci*{EN01} to consider the implications that leverage and asset liquidation strategies can have on the spread of systemic risk.  In particular, we consider the case in which some firm begins a fire sale in order to delever if mark-to-market accounting leaves them violating the leverage constraint.  This generalizes \ci*{feinstein2015illiquid} -- which considers asset liquidation strategies, but not the presence of leverage constraints -- in the direction of \ci*{CFS05} in which fire sales are triggered by capital adequacy requirements.  Mathematically we prove conditions for the existence of a clearing payment and asset prices under either a known liquidation strategy or under the game theoretic strategy of wealth maximization (in which case the liquidation strategy itself becomes part of the equilibrium solution).

We focus the majority of our attention on the numerical case studies in which we calibrate a network model to asset and liability data for large financial institutions within the United States from the years 2007-2014.  We consider multiple parameters that we can vary while staying true to the data.  By varying these parameters in a systematic way, we are able to run counterfactual scenarios and draw broad conclusions about the health of the financial system.  In particular, we find that once the system is begun, a greater maximum leverage requirement leads to fewer fire sales and thus a more stable system.  In contrast, if we create a counterfactual financial system so that the firms begin at the maximum leverage requirement, a stricter leverage requirement leads to fewer assets to be impacted by price fluctuations through mark-to-market accounting.  Thus under this counterfactual system, stricter leverage requirements produce a system more robust to systematic shocks.

\bibliography{bibtex2}
\bibliographystyle{myjmr}

\clearpage
\appendix

\section{Data}\label{Sec:data}
In Section~\ref{Sec:cases} we utilize asset and liability data for 50 large financial institutions in the United States from 2007-2014.  This data is broken down yearly into total assets and total liabilities.  A subset of the collated data is presented in Table~\ref{Table:data}.  The remainder of this section displays information on the leverage ratios for the financial institutions as a group.  Histograms of the realized leverage ratios for each year are displayed in Figures~\ref{Fig:hist_2007}-\ref{Fig:hist_2014}.  Figures~\ref{Fig:scatter_2007}-\ref{Fig:scatter_2014} display the realized leverage ratios against total assets, liabilities, and equity for each firm as scatter plots.  We do this to demonstrate, in general, that there is no discernible relation between firm size and realized leverage ratios.  This can be used to justify our choice of a single leverage requirement for each firm in the case studies of Section~\ref{Sec:cases} (as opposed to leverage requirements that depend on, e.g., initial firm size).

\newcolumntype{H}{>{\iffalse}l<{\fi}@{}}
\begin{sidewaystable}\label{Table:data}
\centering
\begin{tabular}{|l|HcH|*{5}{r|}}
\hline 
\multicolumn{9}{|c|}{Total Assets}\\ \hline
Banks Name&City &State &Parent Name &TotAssets 2007&TotAssets 2008&TotAssets 2009&TotAssets 2010&TotAssets 2011\\ \hline 
JPMorgan Chase Bank, N.A.&Columbus&OH&JPMorgan Chase \& Co.&1,318,888,000&1,746,242,000&1,627,684,000&1,631,621,000&1,811,678,000\\
Bank of America, N.A.&Charlotte&NC&Bank of America Corporation&1,312,794,218&1,470,276,918&1,465,221,449&1,482,278,257&1,459,157,302\\
Wells Fargo Bank, N.A.&Sioux Falls&SD&Wells Fargo \& Company&467,861,000&538,958,000&608,778,000&1,102,278,000&1,161,490,000\\
Citibank, N.A.&Sioux Falls&SD&Citigroup Inc.&1,251,715,000&1,227,040,000&1,161,361,000&1,154,293,000&1,288,658,000\\
U.S. Bank National Association&Cincinnati&OH&U.S. Bancorp&232,759,503&261,775,591&276,376,130&302,259,544&330,470,810\\
PNC Bank, N.A.&Wilmington&DE&PNC Financial Services Group, Inc.&124,782,289&140,777,455&260,309,849&256,638,747&263,309,559\\
Bank of New York Mellon&New York&NY&Bank of New York Mellon Corporation&115,672,000&195,164,000&164,275,000&181,792,000&256,205,000\\
State Street Bank and Trust Company&Boston&MA&State Street Corporation&134,001,964&171,227,778&153,740,526&155,528,576&212,292,942\\
Capital One, N.A.&McLean&VA&Capital One Financial Corporation&97,517,902&115,142,306&127,360,045&126,901,280&133,477,760\\
TD Bank, N.A.&Wilmington&DE&Toronto-Dominion Bank&45,485,973&101,632,075&140,038,551&168,748,912&188,912,554\\ \hline 
\hline
\multicolumn{9}{|c|}{Total Liabilities}\\ \hline
Banks Name&City &State &Parent Name &TotLiab 2007&TotLiab 2008&TotLiab 2009&TotLiab 2010&TotLiab 2011\\ \hline 
JPMorgan Chase Bank, N.A.&Columbus&OH&JPMorgan Chase \& Co.& 1,211,274,000 & 1,616,446,000 & 1,499,365,000 & 1,508,222,000 & 1,680,723,000 \\
Bank of America, N.A.&Charlotte&NC&Bank of America Corporation& 1,202,530,204 & 1,336,942,149 & 1,298,530,760 & 1,302,464,744 & 1,281,297,032 \\
Wells Fargo Bank, N.A.&Sioux Falls&SD&Wells Fargo \& Company& 426,089,000 & 497,153,000 & 552,185,000 & 978,716,000 & 1,036,999,000 \\
Citibank, N.A.&Sioux Falls&SD&Citigroup Inc.& 1,151,143,000 & 1,144,600,000 & 1,043,468,000 & 1,026,333,000 & 1,136,309,000 \\
U.S. Bank National Association&Cincinnati&OH&U.S. Bancorp& 210,012,660 & 238,925,837 & 250,150,518 & 271,432,471 & 295,803,555 \\
PNC Bank, N.A.&Wilmington&DE&PNC Financial Services Group, Inc.& 109,846,056 & 127,862,023 & 228,481,932 & 222,888,583 & 227,956,188 \\
Bank of New York Mellon&New York&NY&Bank of New York Mellon Corporation& 106,543,000 & 183,441,000 & 150,538,000 & 165,927,000 & 237,946,000 \\
State Street Bank and Trust Company&Boston&MA&State Street Corporation& 122,060,508 & 157,881,228 & 139,072,508 & 138,831,760 & 193,568,478 \\
Capital One, N.A.&McLean&VA&Capital One Financial Corporation& 77,436,577 & 95,108,980 & 103,905,623 & 102,673,731 & 108,780,941 \\
TD Bank, N.A.&Wilmington&DE&Toronto-Dominion Bank& 35,794,451 & 83,340,196 & 117,537,086 & 142,907,066 & 161,094,419 \\ \hline
\end{tabular}
\caption{A sample of the data of assets and liabilities for US banks.}
\end{sidewaystable}

\begin{figure}
\centering
\includegraphics[height=0.4\textheight,width=0.8\textwidth]{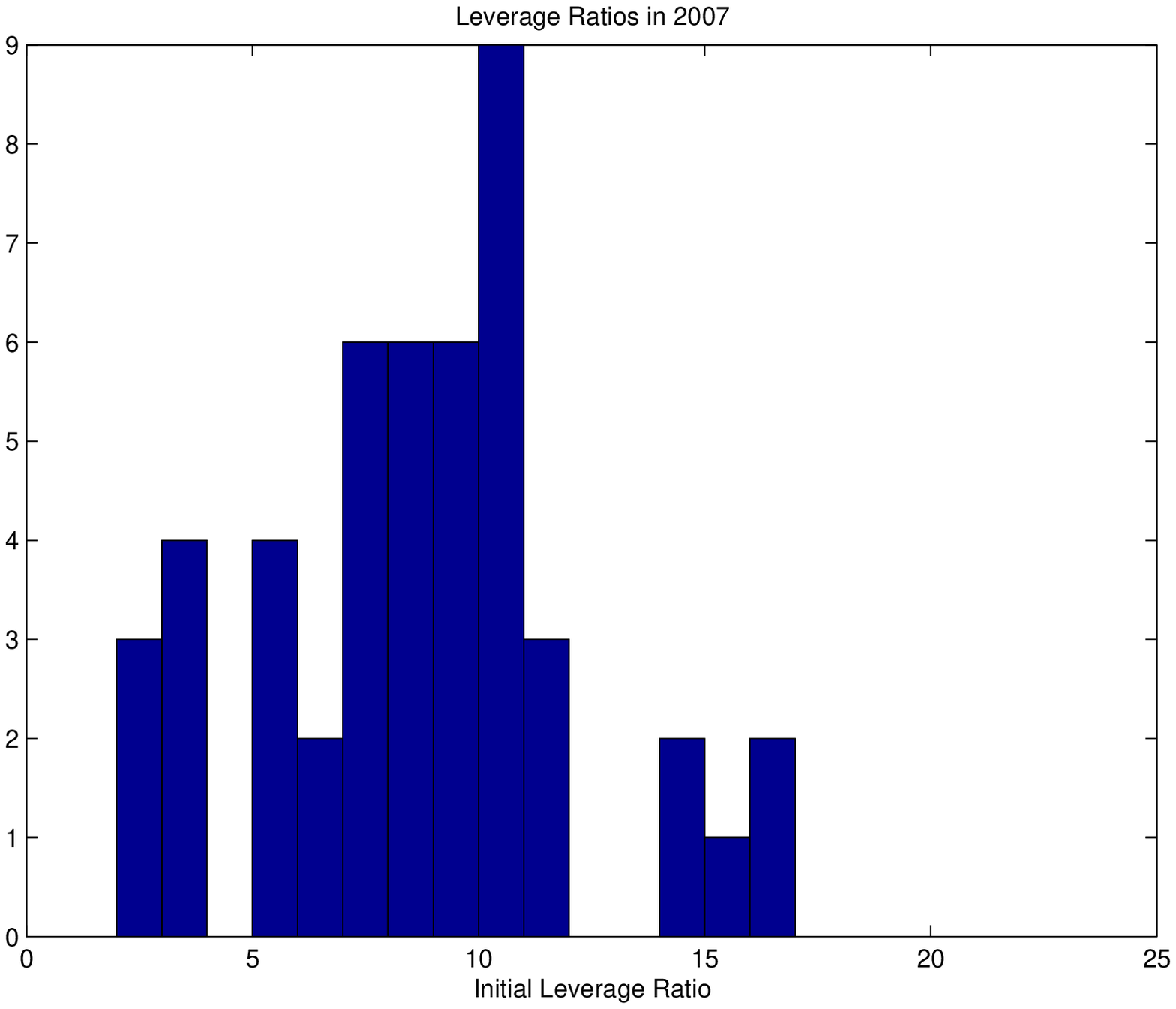}
\caption{A histogram of the leverage ratios from the initial data set on US banks in 2007.}
\label{Fig:hist_2007}
\end{figure}

\begin{figure}
\centering
\includegraphics[height=0.4\textheight,width=0.8\textwidth]{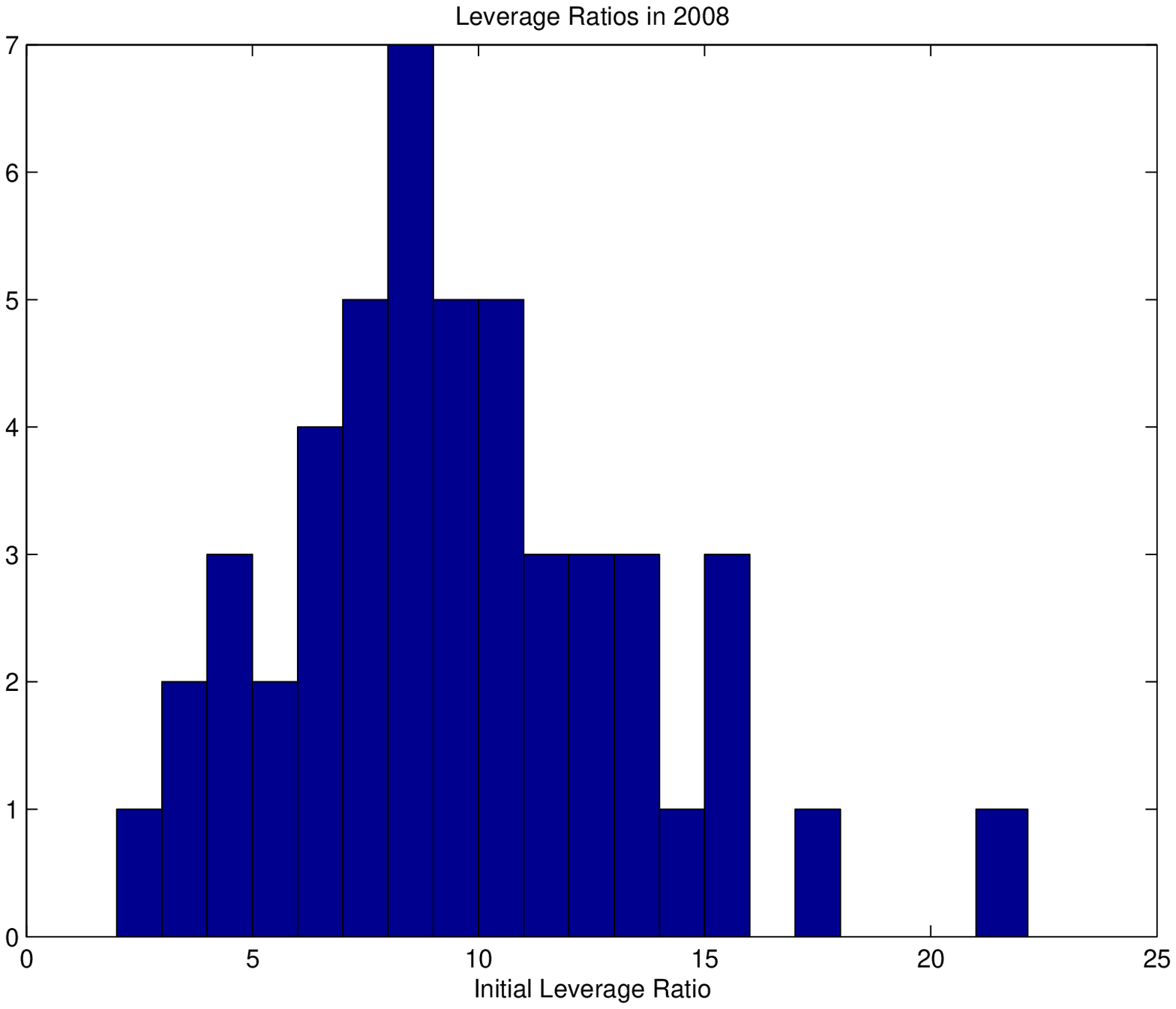}
\caption{A histogram of the leverage ratios from the initial data set on US banks in 2008.}
\label{Fig:hist_2008}
\end{figure}

\begin{figure}
\centering
\includegraphics[height=0.4\textheight,width=0.8\textwidth]{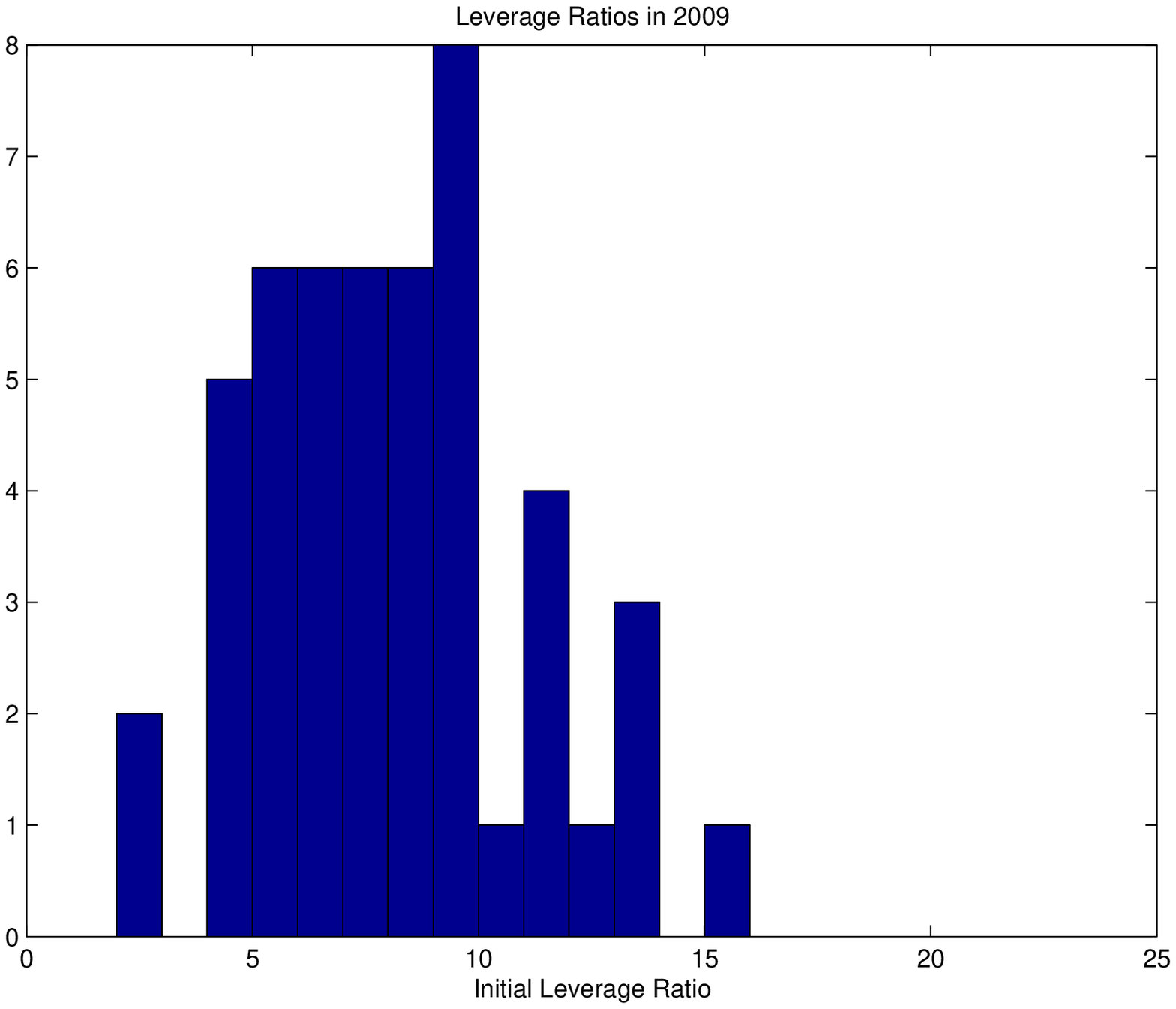}
\caption{A histogram of the leverage ratios from the initial data set on US banks in 2009.}
\label{Fig:hist_2009}
\end{figure}

\begin{figure}
\centering
\includegraphics[height=0.4\textheight,width=0.8\textwidth]{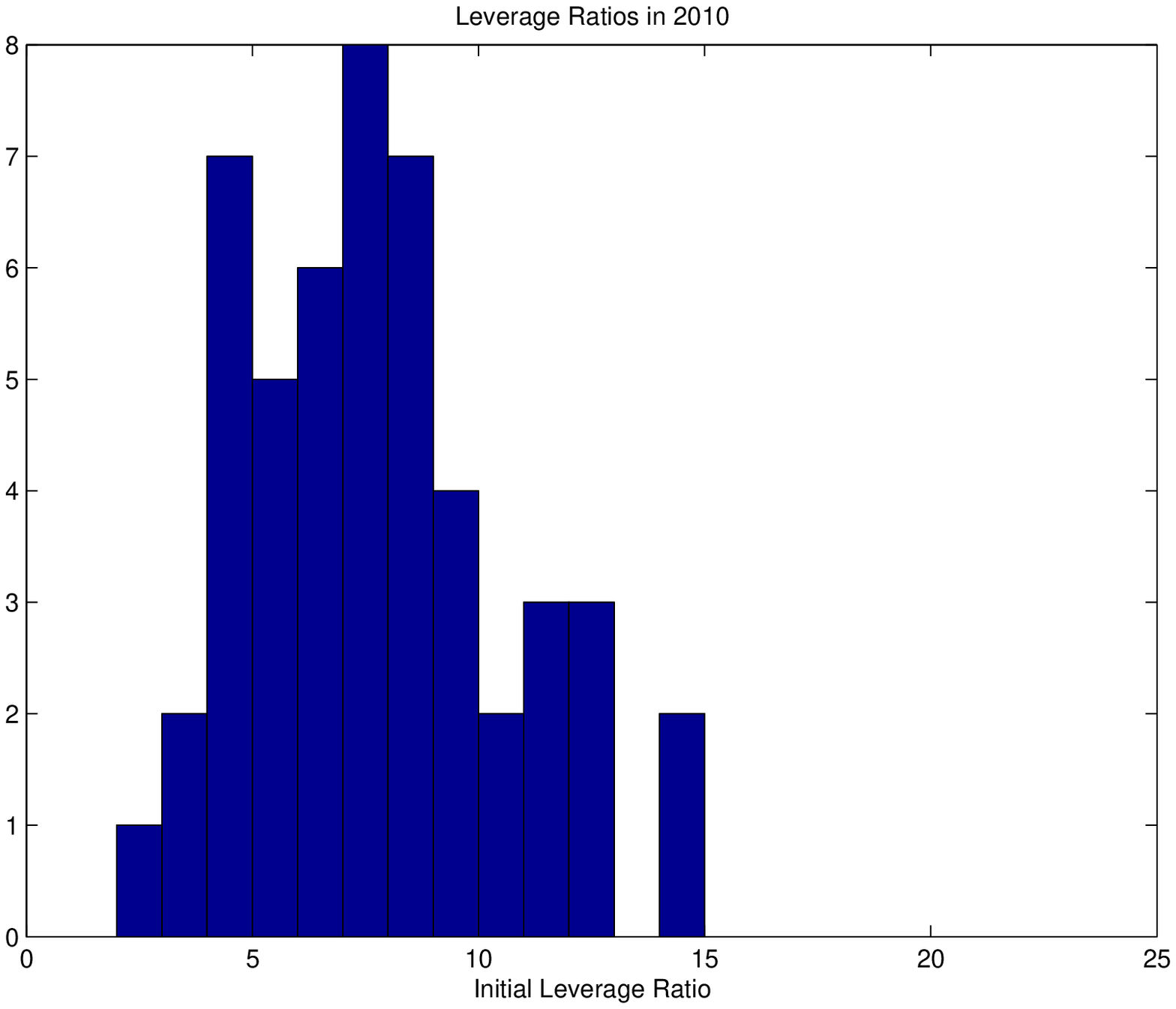}
\caption{A histogram of the leverage ratios from the initial data set on US banks in 2010.}
\label{Fig:hist_2010}
\end{figure}

\begin{figure}
\centering
\includegraphics[height=0.4\textheight,width=0.8\textwidth]{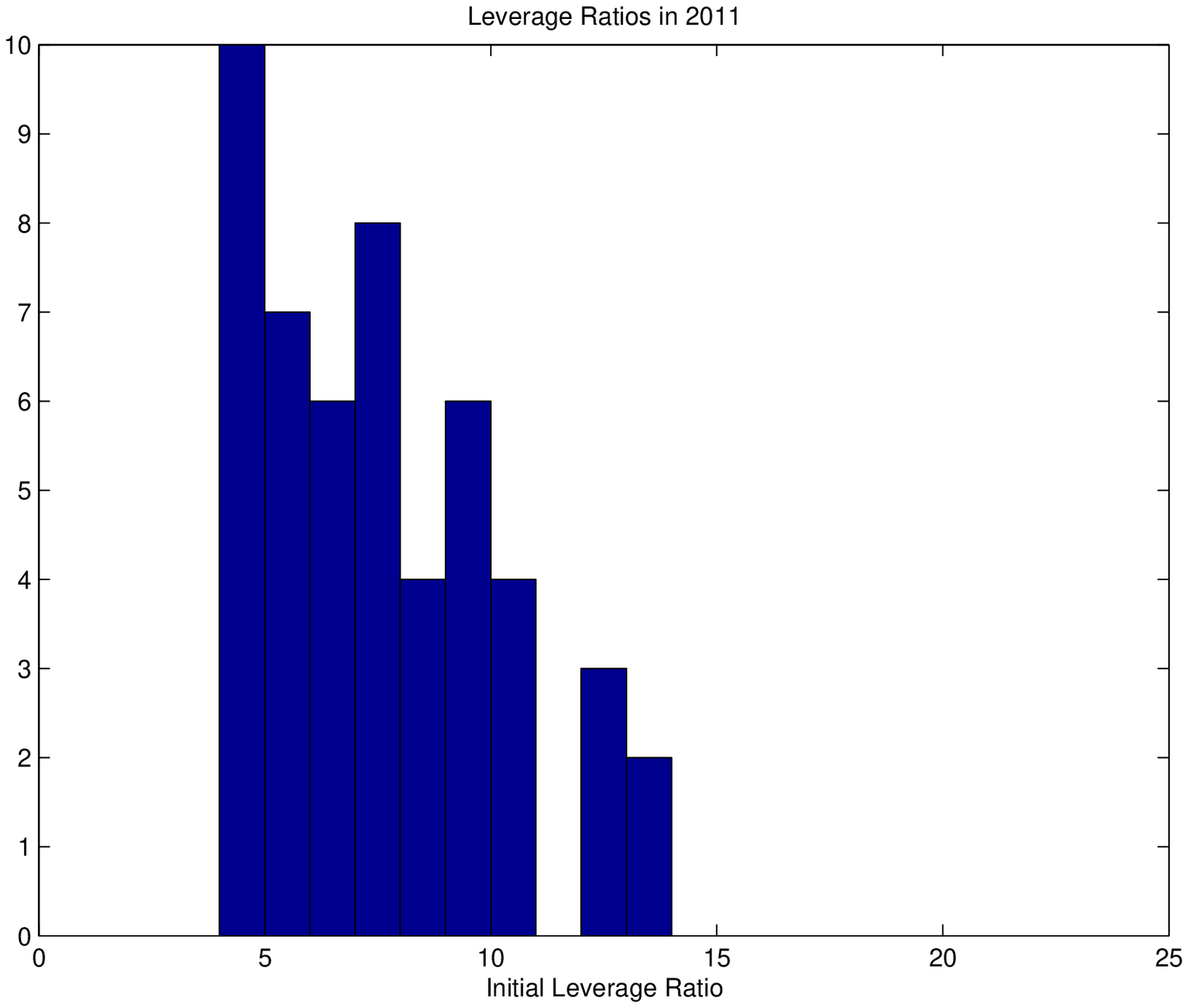}
\caption{A histogram of the leverage ratios from the initial data set on US banks in 2011.}
\label{Fig:hist_2011}
\end{figure}

\begin{figure}
\centering
\includegraphics[height=0.4\textheight,width=0.8\textwidth]{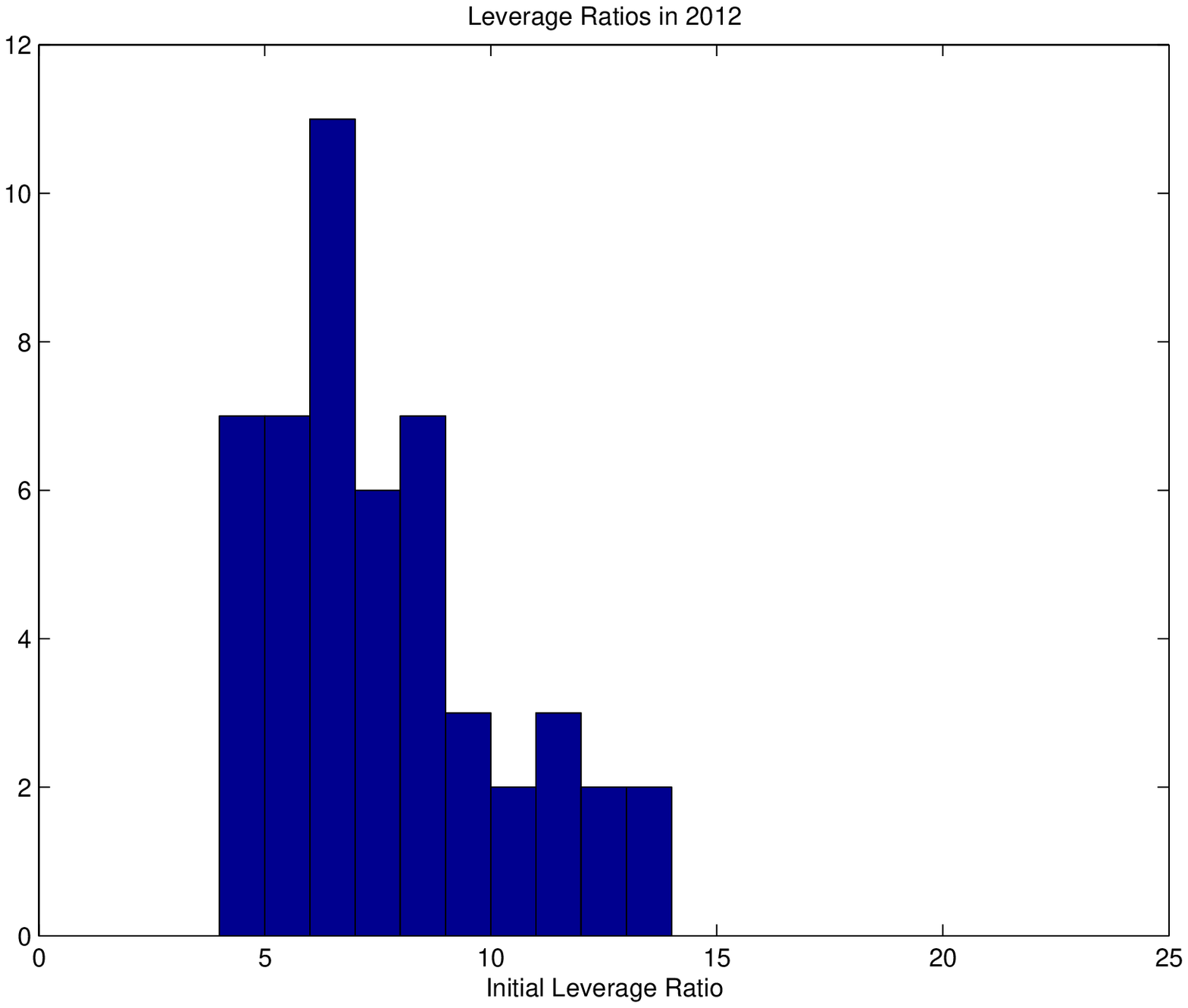}
\caption{A histogram of the leverage ratios from the initial data set on US banks in 2012.}
\label{Fig:hist_2012}
\end{figure}

\begin{figure}
\centering
\includegraphics[height=0.4\textheight,width=0.8\textwidth]{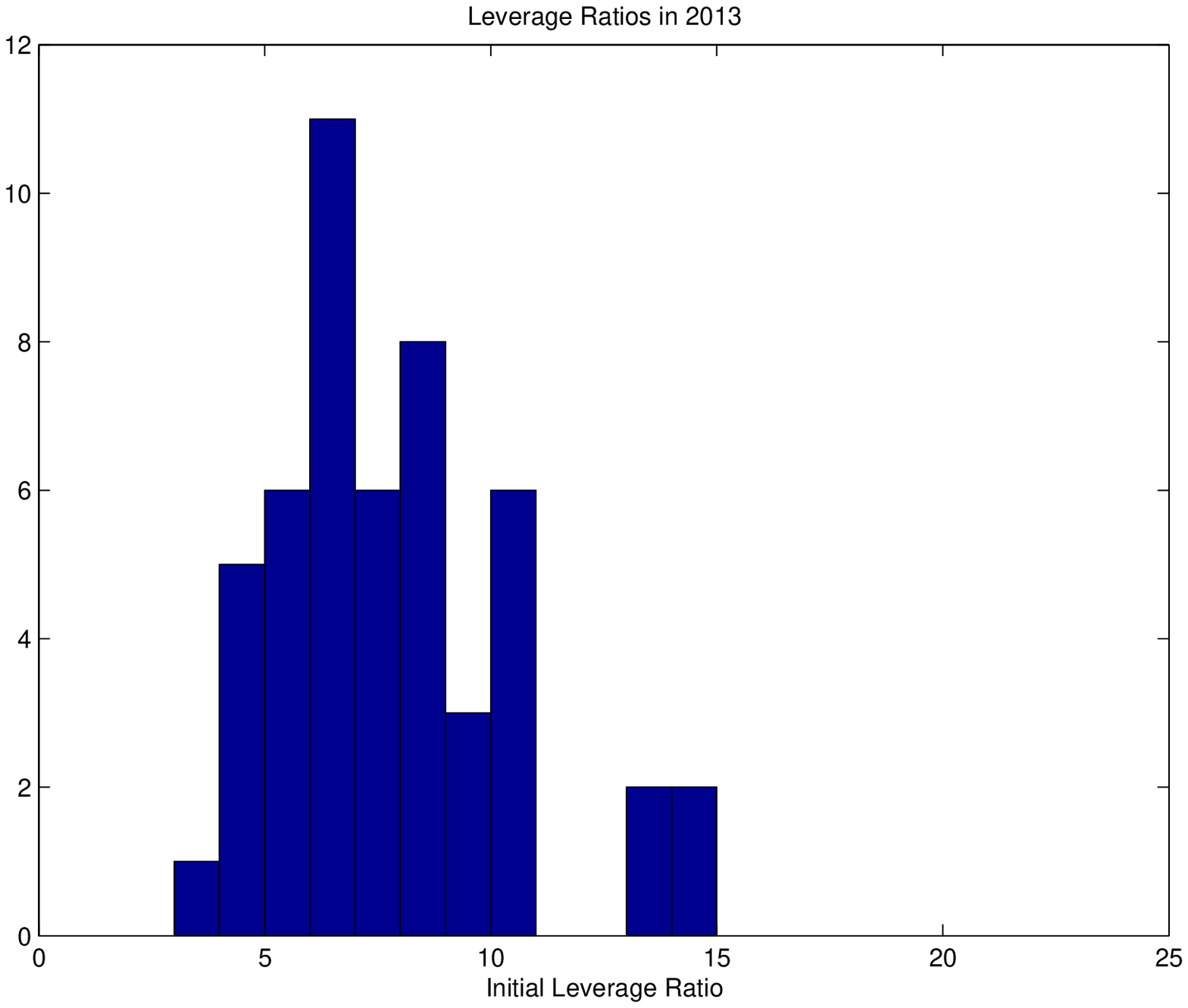}
\caption{A histogram of the leverage ratios from the initial data set on US banks in 2013.}
\label{Fig:hist_2013}
\end{figure}

\begin{figure}
\centering
\includegraphics[height=0.4\textheight,width=0.8\textwidth]{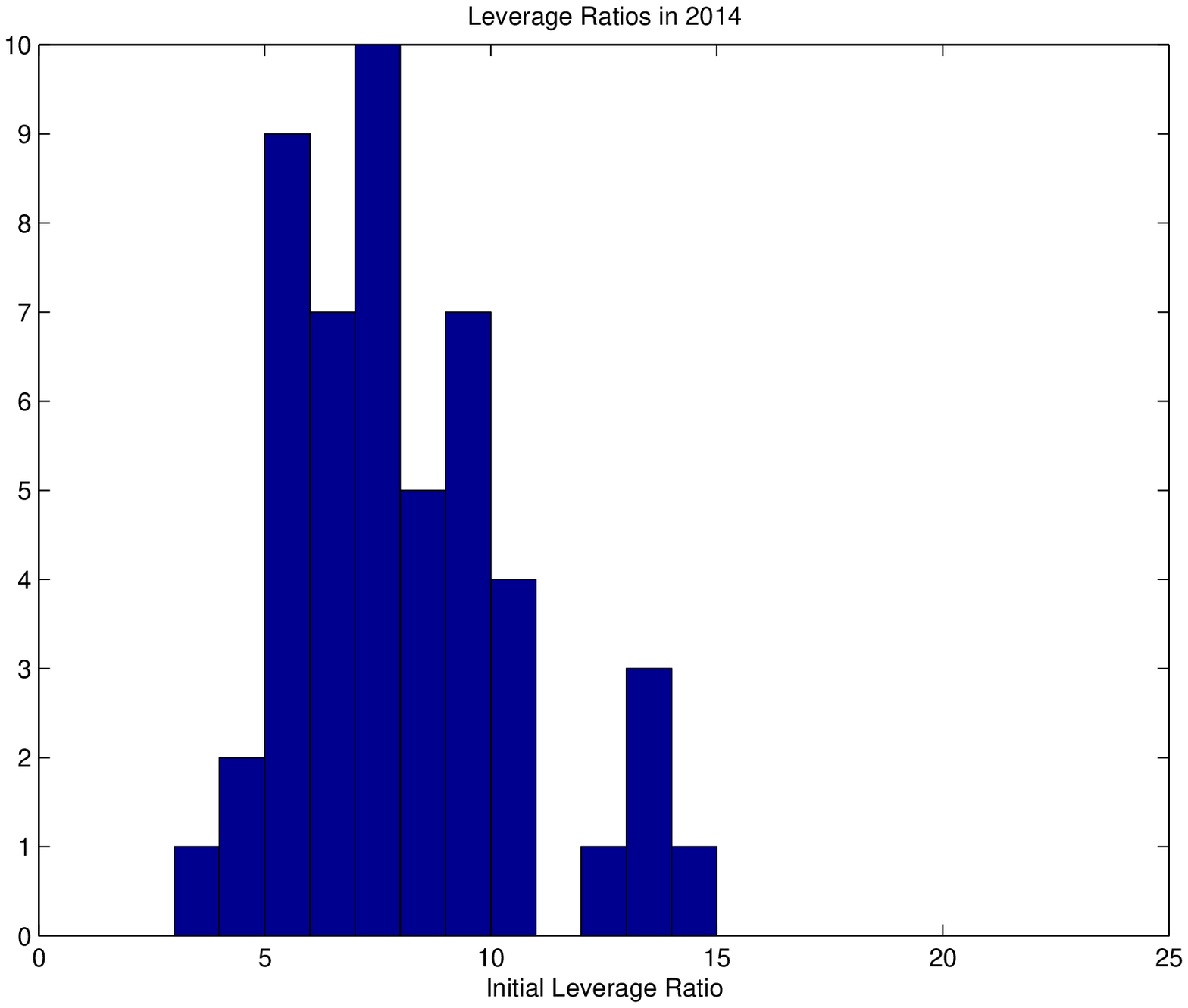}
\caption{A histogram of the leverage ratios from the initial data set on US banks in 2014.}
\label{Fig:hist_2014}
\end{figure}

\begin{figure}
\centering
\includegraphics[height=0.4\textheight,width=0.8\textwidth]{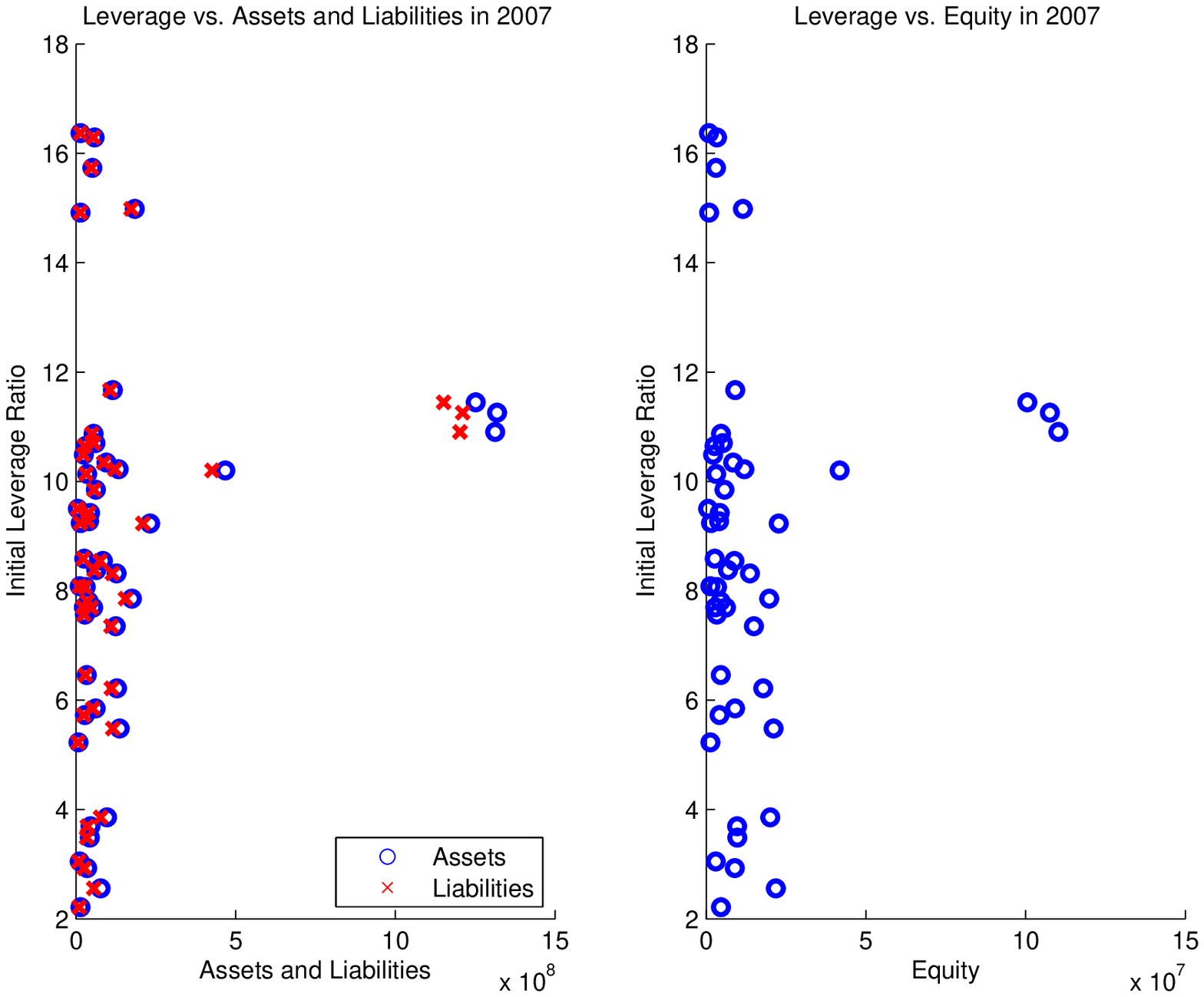}
\caption{A plot of leverage ratios as a function of assets, liabilities, and equity in 2007.}
\label{Fig:scatter_2007}
\end{figure}

\begin{figure}
\centering
\includegraphics[height=0.4\textheight,width=0.8\textwidth]{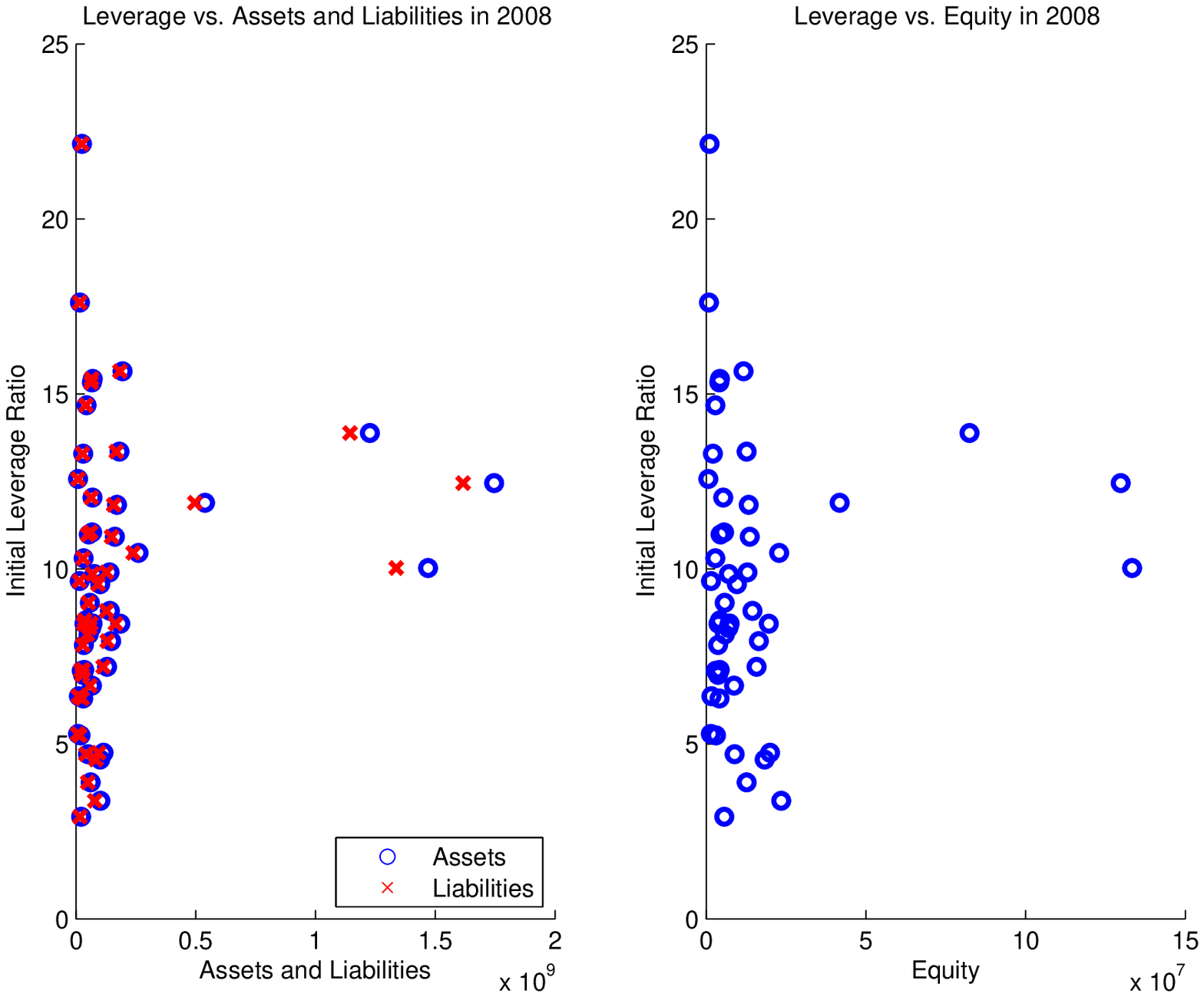}
\caption{A plot of leverage ratios as a function of assets, liabilities, and equity in 2008.}
\label{Fig:scatter_2008}
\end{figure}

\begin{figure}
\centering
\includegraphics[height=0.4\textheight,width=0.8\textwidth]{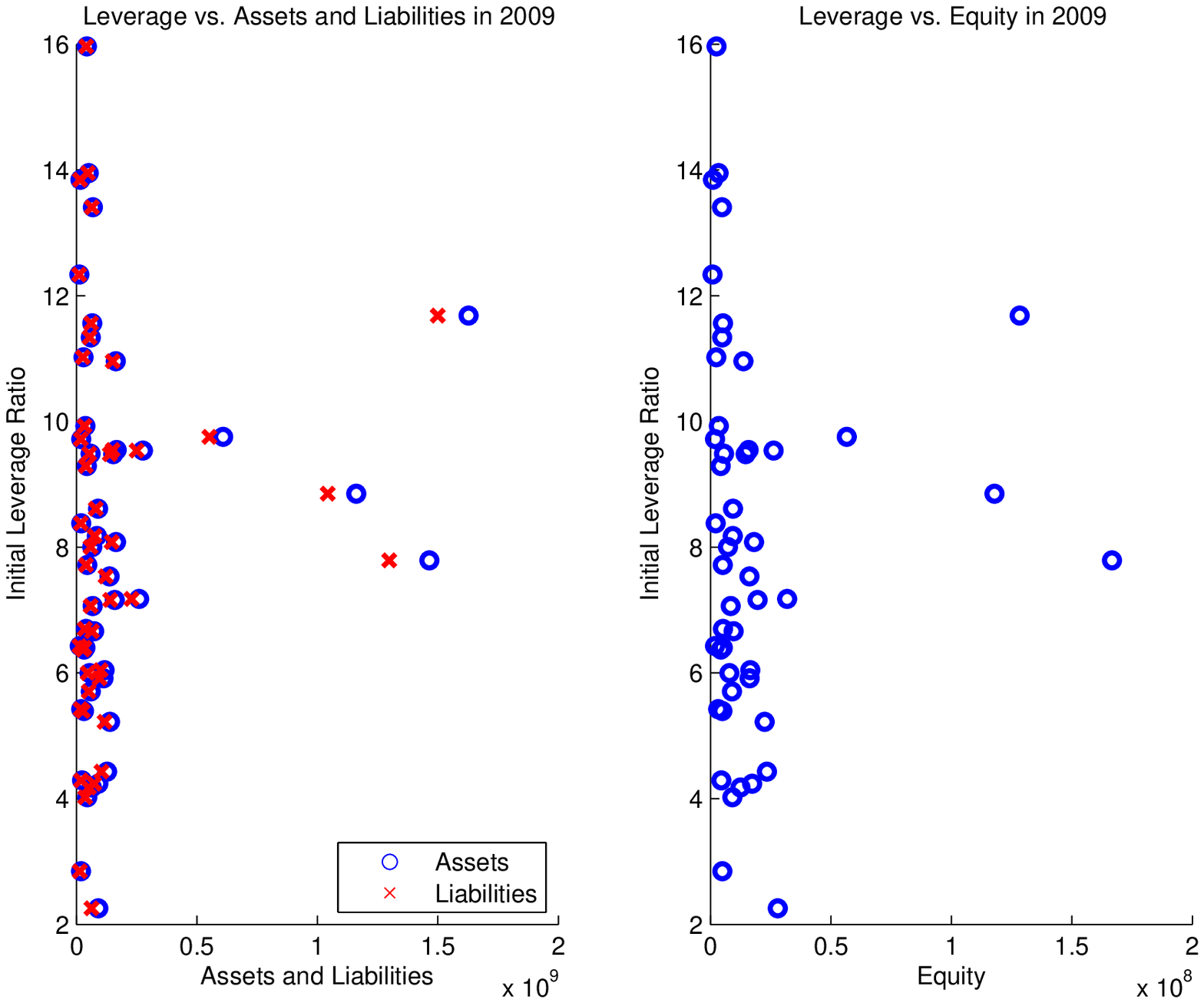}
\caption{A plot of leverage ratios as a function of assets, liabilities, and equity in 2009.}
\label{Fig:scatter_2009}
\end{figure}

\begin{figure}
\centering
\includegraphics[height=0.4\textheight,width=0.8\textwidth]{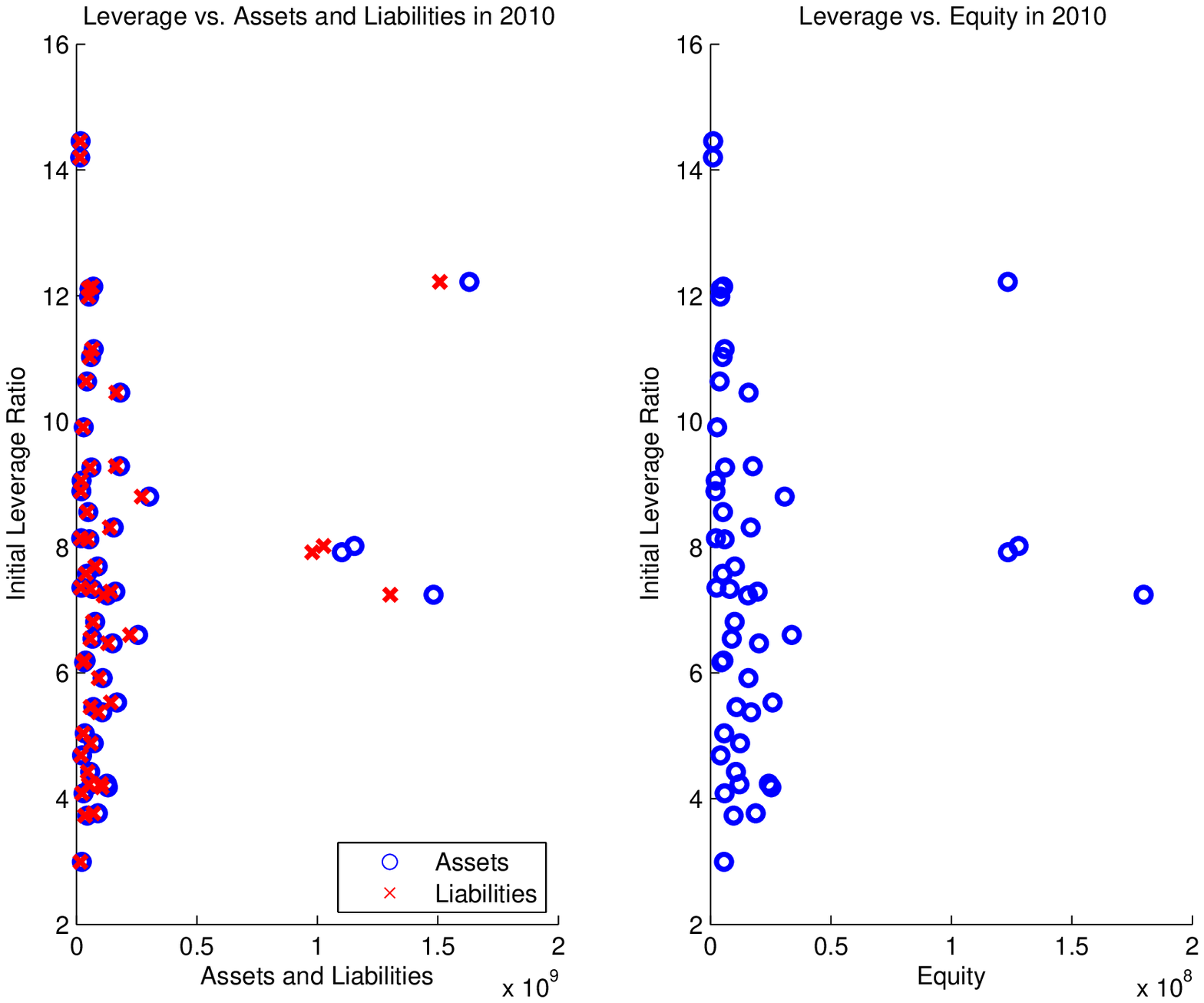}
\caption{A plot of leverage ratios as a function of assets, liabilities, and equity in 2010.}
\label{Fig:scatter_2010}
\end{figure}

\begin{figure}
\centering
\includegraphics[height=0.4\textheight,width=0.8\textwidth]{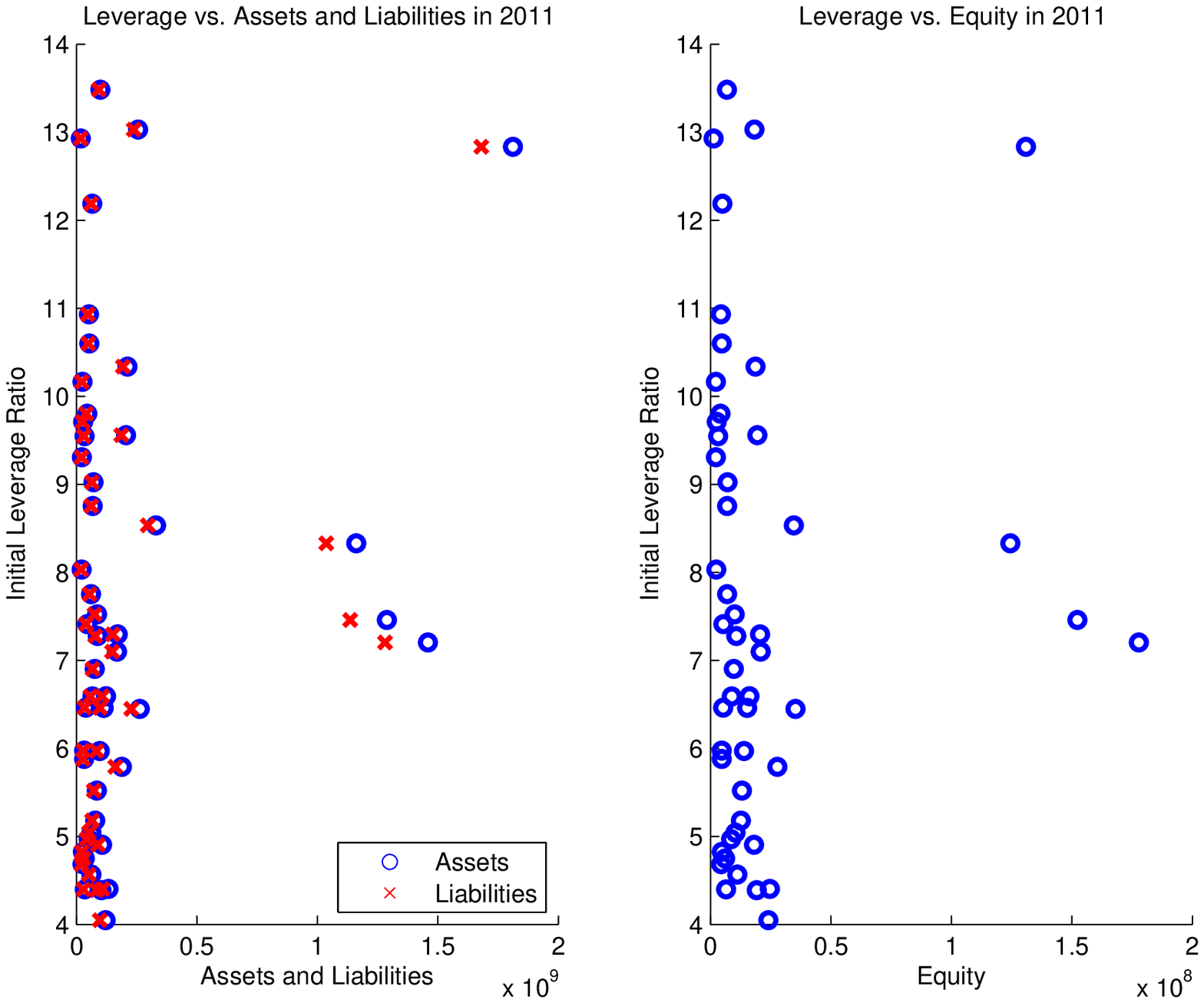}
\caption{A plot of leverage ratios as a function of assets, liabilities, and equity in 2011.}
\label{Fig:scatter_2011}
\end{figure}

\begin{figure}
\centering
\includegraphics[height=0.4\textheight,width=0.8\textwidth]{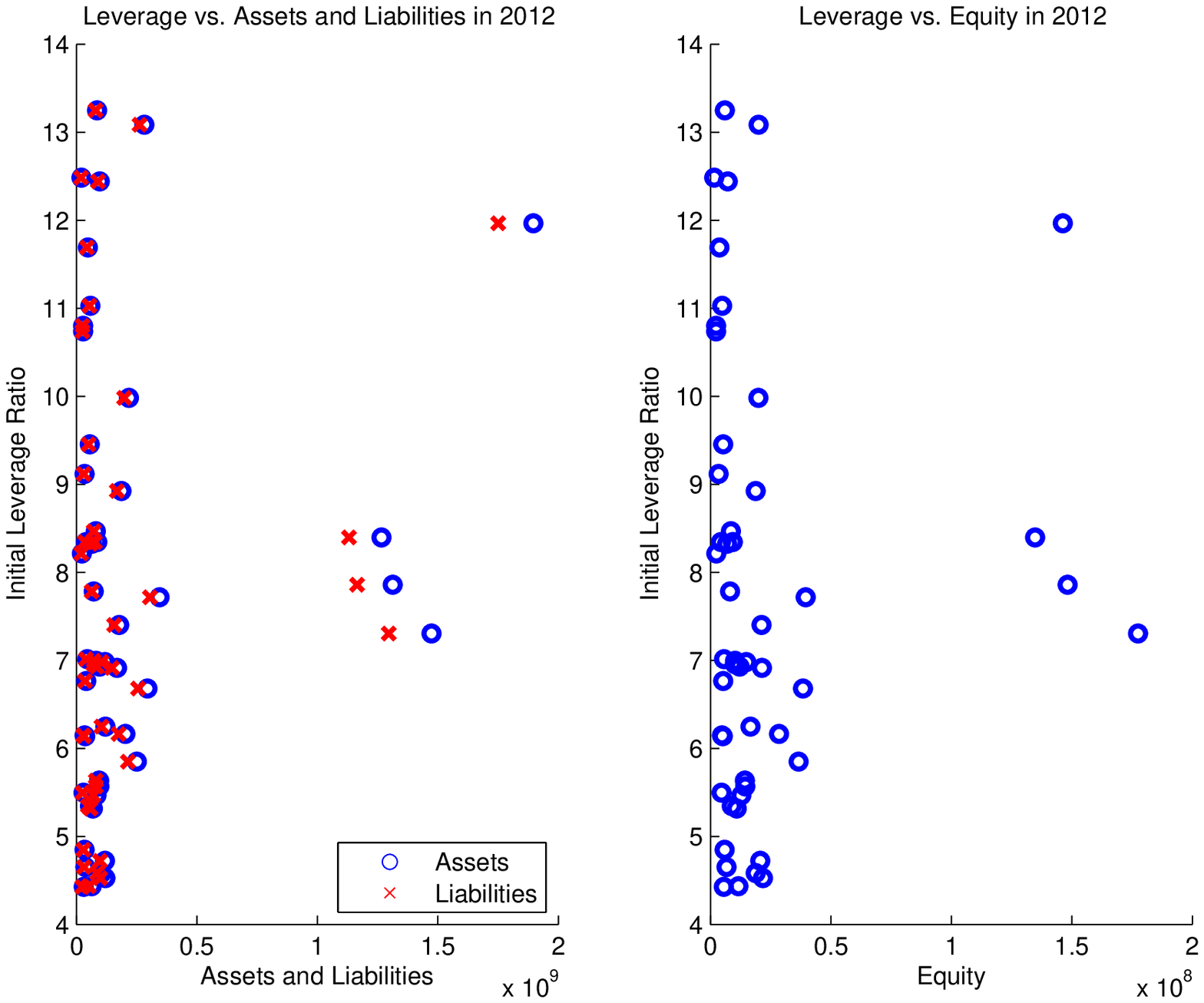}
\caption{A plot of leverage ratios as a function of assets, liabilities, and equity in 2012.}
\label{Fig:scatter_2012}
\end{figure}

\begin{figure}
\centering
\includegraphics[height=0.4\textheight,width=0.8\textwidth]{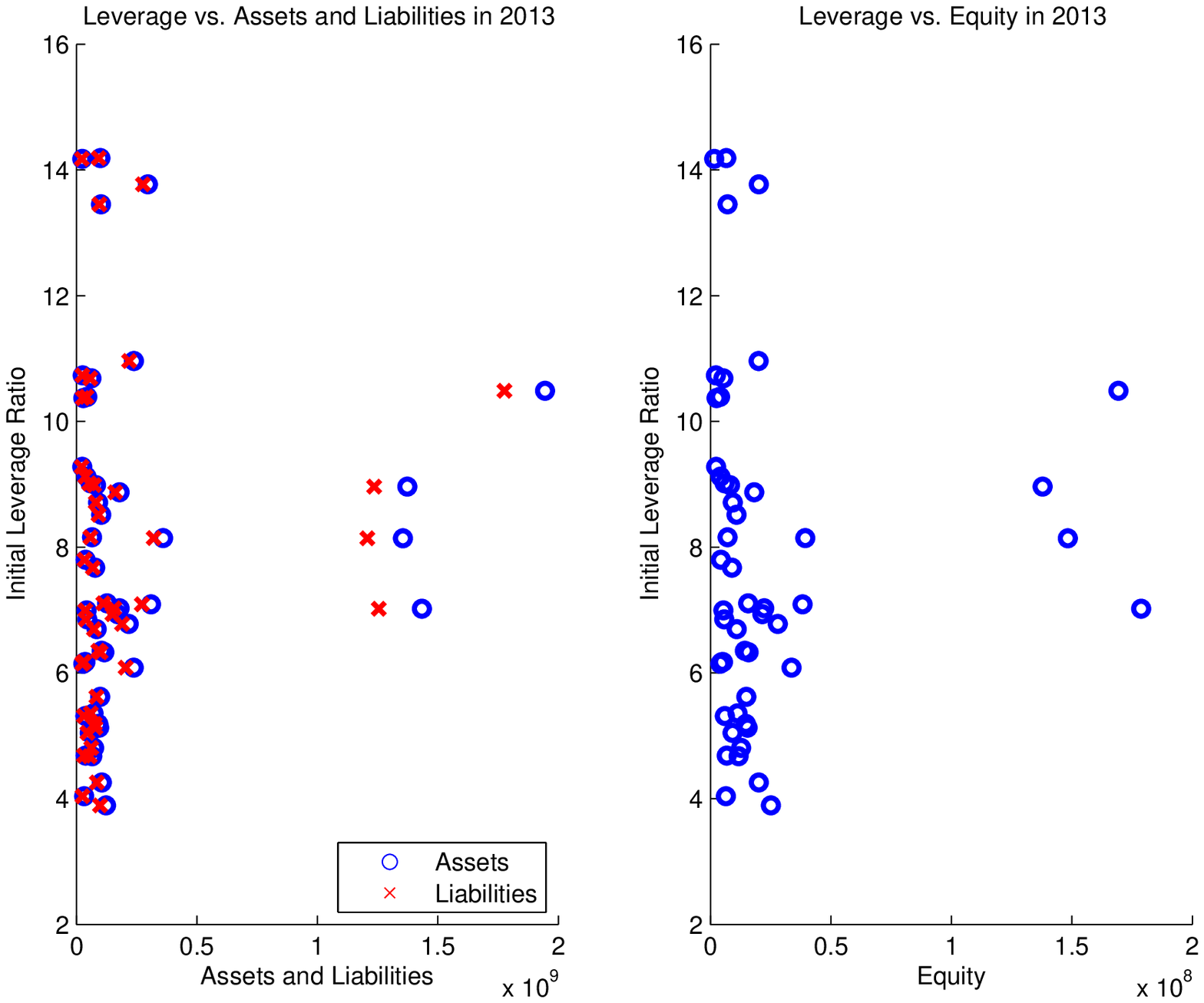}
\caption{A plot of leverage ratios as a function of assets, liabilities, and equity in 2013.}
\label{Fig:scatter_2013}
\end{figure}

\begin{figure}
\centering
\includegraphics[height=0.4\textheight,width=0.8\textwidth]{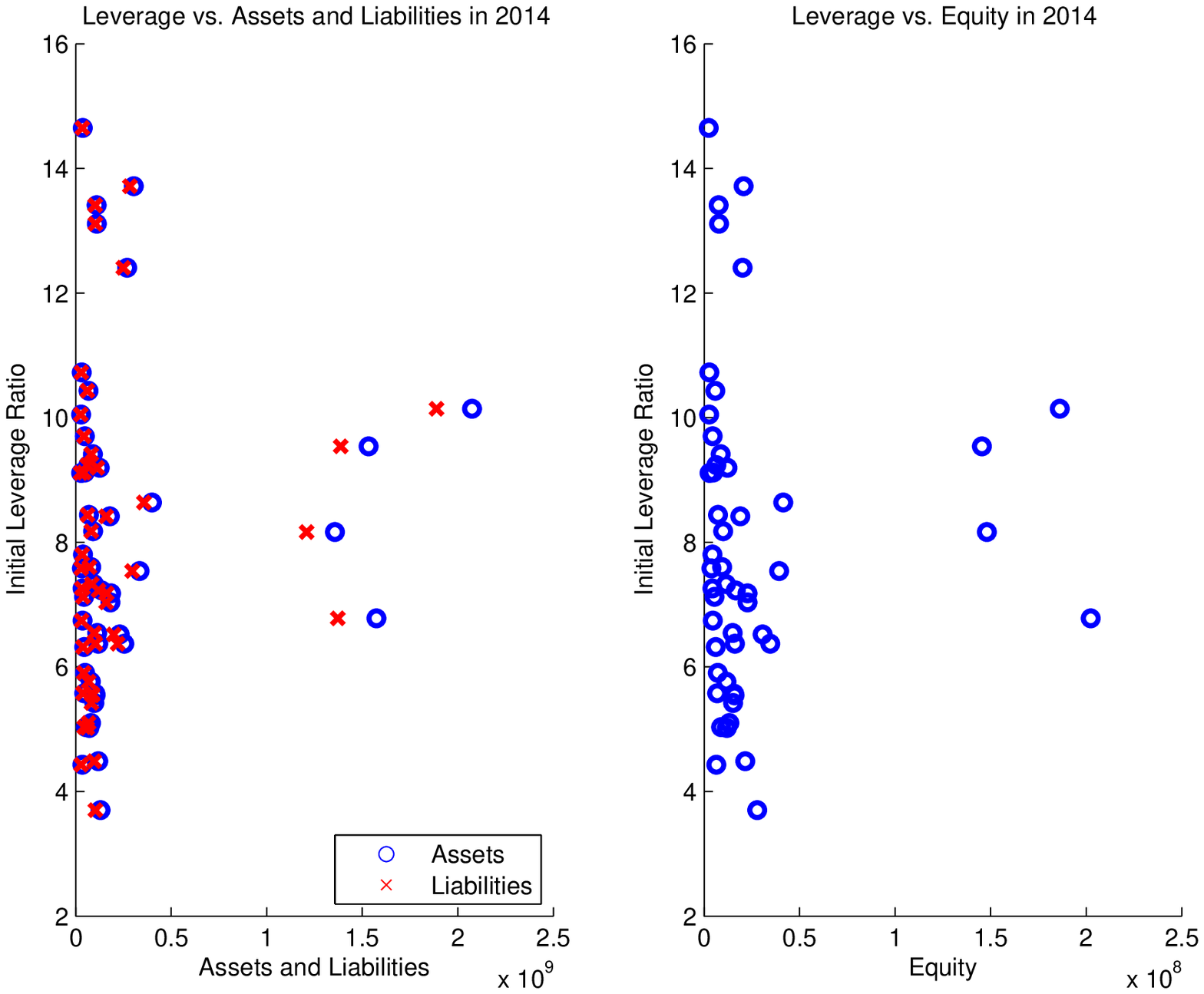}
\caption{A plot of leverage ratios as a function of assets, liabilities, and equity in 2014.}
\label{Fig:scatter_2014}
\end{figure}

\end{document}